\documentclass[twoside,12pt,reqno]{amsart}
\usepackage{a4wide}
\usepackage{latexsym}
\usepackage{amsfonts}
\usepackage{amssymb}
\usepackage{amsthm}
\usepackage{amsmath}              
\usepackage{wasysym}              
\newcommand{\cl}{C \kern -0.1em \ell}     

\newcommand{\JJ}{\mathbin{\raisebox{0.25ex}{$\footnotesize
                       \rm\vphantom{I}%
                       \_\hskip -0.25em\_%
                       \vrule width 0.6pt$}}}           

\newcommand{\w}{\wedge}

\newcommand{\Spin}{\mathbf{Spin}}            
\newcommand{\Pin}{\mathbf{Pin}}

\newcommand{\Mat}{{\rm Mat}}

  \newcommand{\be}{{\bf e}}

\newcommand{\e}{{\bf e}}

\newcommand{\Id}{{\bf 1}}

\newcommand{\BK}{\mathbb{K}}

\newcommand{\BC}{\mathbb{C}}
\newcommand{\BR}{\mathbb{R}}
\newcommand{\BH}{\mathbb{H}}
\newcommand{\BZ}{\mathbb{Z}}
\newcommand{\diag}{\mbox{\rm diag}}
\newcommand{\spn}{\mbox{\rm span}}

\def\dim{\hbox{\rm dim\,}}

\newcommand{\beq}{\begin{equation}}
\newcommand{\eeq}{\end{equation}}

\def\dim{\hbox{\rm dim\,}}

\def\CLIFFORD{\mbox{\bf \tt CLIFFORD}}

\newcommand{\ed}{\end{document}}
\def\ve{\varepsilon}

\def\Sg{{\hat{S}}}
\def\fg{\hat{f}}
\def\psig{{\psi_g}}
\def\phig{{\phi_g}}
\def\chpsi{\check{\psi}}
\def\chphi{\check{\phi}}

\newcommand{\ta}[2]{#1_#2\tilde{\phantom{.}}}
\newcommand{\cb}[1]{\mathcal{#1}}

\newcommand{\Gpqe}[2]{G_{#1,#2}^\varepsilon}
\newcommand{\Gpq}[2]{G_{#1,#2}}
\newcommand{\Gpqf}[3]{G_{#1,#2}(#3)}

\newcommand{\Kpqf}[3]{K_{#1,#2}(#3)}
\newcommand{\Tpqf}[3]{T_{#1,#2}(#3)}
\newcommand{\Lip}[2]{\mathrm{\bf Lip}(#1,#2)}

\newcommand{\iu}{{\underline{i}}}

\newcommand{\fpower}[1]{{}^2 \kern -0.01em #1}

\newcommand{\chS}{\check{S}}
\newcommand{\chK}{\check{\BK}}
\renewcommand{\Id}{\mathrm{Id}}
\newcommand{\tr}{\mathrm{tr}}

\newcommand{\rexpansion}[3]{#1_1 #2_1 + \ldots + #1_#3 #2_#3}
\newcommand{\cexpansion}[3]{\overline{#1}_1 #2_1 + \ldots + \overline{#1}_#3 #2_#3}

\newcommand{\sympsiphi}[2]{\psi_{(#1}\phi_{#2)}}
\newcommand{\asympsiphi}[2]{\psi_{[#1}\phi_{#2]}}
\newcommand{\colonequal}{\mathrel{\mathop:}=}
\newcommand{\starc}[1]{#1^*}

\newcommand{\tp}{\ta{T}{\ve}}

\newcommand{\phm}{\phantom{-}}

\DeclareMathOperator{\hotimes}{\Hat{\otimes}}

\theoremstyle{plain}
\newtheorem{theorem}{Theorem}

\newtheorem{corollary}{Corollary}
\newtheorem{lemma}{Lemma}
\newtheorem{proposition}{Proposition}
\theoremstyle{definition}
\newtheorem{definition}{Definition}
\newtheorem{example}{Example}
\newtheorem*{remark}{Remark}

\begin{document}
\title[On the transposition anti-involution III]{On the Transposition Anti-Involution in\\
       Real Clifford Algebras III:\\
       The Automorphism Group of the Transposition\\
       Scalar Product on Spinor Spaces}

\author{Rafa\l \ Ab\l amowicz}
\email{rablamowicz@tntech.edu}
\address{%
Department of Mathematics, Box 5054\\
Tennessee Technological University\\
Cookeville, TN 38505, USA}
\author{Bertfried Fauser}
\email{B.Fauser@cs.bham.ac.uk}
\address{%
School of Computer Science\\
The University of Birmingham\\
Edgbaston-Birmingham, W. Midlands, B15 2TT\\
United Kingdom}

\begin{abstract} 
  A signature $\ve=(p,q)$ dependent transposition anti-involution $\tp$ of
  real Clifford algebras $\cl_{p,q}$ for non-degenerate quadratic forms
  was introduced in~\cite{part1}. In~\cite{part2} we showed that, depending
  on the value of $(p-q) \bmod 8$, the map $\tp$ gives rise to transposition,
  complex Hermitian, or quaternionic Hermitian conjugation of representation
  matrices in spinor representation. The resulting scalar product is in general
  different from the two known standard scalar products~\cite{lounesto}.
  We provide a full signature $(p,q)$ dependent classification of the
  invariance groups $\Gpqe{p}{q}$ of this product for $p+q\leq 9.$ The map~$\tp$ 
  is identified as the ``star" map known~\cite{passman} from the theory of (twisted) 
  group algebras, where the Clifford algebra $\cl_{p,q}$ is seen as a twisted group
  ring $\BR^t[(\BZ_2)^n],$ $n=p+q.$ We discuss and list important subgroups of
  stabilizer groups $\Gpqf{p}{q}{f}$ and their transversals in relation to
  generators of spinor spaces. 
\end{abstract}



\maketitle
\medskip\section{Introduction}
\label{intro}
This article concludes the developments of~\cite{part1,part2}. Hence, all
notation and definitions are the same. Recall that in addition to reversion,
grade involution and conjugation, any universal real Clifford algebra
$\cl(V,Q)$ of a non-degenerate quadratic real vector space possesses a
transposition anti-involution $\tp$ as a unique extension of a certain
orthogonal map $t_\varepsilon:V \rightarrow V^\flat$. Using the identification
$V^\flat \cong V^\ast,$ the orthogonal map $t_\varepsilon$ was shown to be a
symmetric non-degenerate correlation on~$V$. For the properties of 
$t_\varepsilon$ see~\cite{part1}.

Let $\cl_{p,q}$ be a simple Clifford algebra ($p-q \neq 1 \bmod 4$) and let 
$S = \cl_{p,q}f$ be a left spinor ideal generated by a primitive idempotent
$f$. As usual, we let $\BK = f\cl_{p,q}f.$ Then, $\tp$ allows one to
define~\cite[Prop. 3]{part2} a $\BK$-valued \emph{transposition spinor scalar 
product} $S \times S \rightarrow \BK$ as\footnote{%
Recall that elements $\lambda$ of $\BK=f\cl_{p,q}f$ commute with $f$ because
$\lambda = fuf$ for some $u \in \cl_{p,q}.$ The notion $\exists !$ stands for
`there exists a unique'.}
\begin{gather}
  \forall \psi,\phi \in S, \quad \exists ! \, \lambda \in \BK,\quad\textrm{s.t.}\quad
  (\psi,\phi) \mapsto \tp(\psi)\phi = \lambda f = f\lambda .
\label{eq:norm}
\end{gather}
The group $\Gpqe{p}{q}$ was defined in~\cite{part2} as the invariance group
of this transposition scalar product:
\begin{gather}
  \Gpqe{p}{q} = \{ g \in \cl_{p,q} \, | \, \tp(g)g =1\}.
\label{eq:Gpq}
\end{gather}
In particular, the product is invariant under two of its subgroups:
Salingaros'~\cite{salingaros1,salingaros2,salingaros3} finite vee
group $\Gpq{p}{q} < \Gpqe{p}{q}$ and the stabilizer group
$\Gpqf{p}{q}{f} \lhd \Gpq{p}{q}$ of a primitive idempotent $f$ under
the conjugate action of $\Gpq{p}{q}$ on $\cl_{p,q}.$ Since the stabilizer
group $\Gpqf{p}{q}{f}$ plays an important role in constructing and
understanding spinor representation of Clifford algebras, we classified
these groups in~\cite{part2}. 

In this article, in Section~\ref{subgroups}, we discuss important subgroups
of $\Gpqf{p}{q}{f}$ which are related to the idempotent $f$ and to the
division ring $\BK.$ Furthermore, we show how transversals of these subgroups
in the vee group $\Gpq{p}{q}$ are in fact generators of the (left) spinor
spaces ${}_{\cl}S_{\BK}$ viewed as real vector spaces or as right
modules over $\BK$.

Then, in Section~\ref{tpscalarproduct}, we undertake a systematic study of
the transposition scalar product $\tp(\psi)\phi$ and its automorphism group
$\Gpqe{p}{q}$ for all signatures $(p,q)$, $p+q \leq 9,$ in simple and
semisimple Clifford algebras. 

In Section~\ref{groupalgebra} we conclude that our transposition
anti-involution $\tp$  of $\cl_{p,q}$ is related to a certain
anti-involution $*:\BR^t[(\BZ_2)^n] \rightarrow \BR^t[(\BZ_2)^n]$ when we
identify the Clifford algebra $\cl_{p,q}$ with a \textit{twisted group ring}
$\BR^t[(\BZ_2)^n]$ of an abelian group
$(\BZ_2)^n = \BZ_2 \times \cdots \times \BZ_2$, $n$~times where $n=p+q,$ in
the sense of~\cite{albuquerquemajid}. 

Finally, in Section~\ref{conc}, we summarize the results obtained in all
three papers and provide some ideas how they can be extended and generalized.
 
\medskip\section{Subgroups of the stabilizer group}
\label{subgroups}

Let $\cl_{p,q}$ be a real universal Clifford algebra with Grassmann basis
$\cb{B}$ sorted by the admissible order \texttt{InvLex}~\cite{part1}. We
begin by recalling the definition of Salingaros' vee
group~\cite{salingaros1}.
\begin{definition}
A \textit{vee group} $\Gpq{p}{q}$ is defined as the set $\Gpq{p}{q} = 
\{ \pm m \mid m \in \cb{B} \}$ in $\cl_{p,q}$ together with the Clifford
product as the group binary operation. 
\label{def:veegroup}
\end{definition}
\noindent
This finite group is of order $|\Gpq{p}{q}| = 2^{1+p+q}$ and its commutator
subgroup  $\Gpq{p}{q}' = \{\pm 1\}.$

We know that any polynomial $f$ in $\cl_{p,q}$ expressible as a product of
idempotents $\frac12(1\pm \be_{\iu_s})$ namely
\begin{equation}
  f = \frac12(1\pm \be_{\iu_1})\frac12(1\pm \be_{\iu_2})\cdots
      \frac12(1\pm \be_{\iu_k}),
\label{eq:f}
\end{equation}
where $\be_{\iu_1},\ldots, \be_{\iu_k},$ $k = q - r_{q-p}$,\footnote{%
Here, $r_i$ is Radon-Hurwitz number defined by recursion as $r_{i+8}=r_i+4$
and these initial values: $r_0=0,r_1=1,r_2=r_3=2,r_4=r_5=r_6=r_7=3
$~\cite{hahn,lounesto}.}
are commuting basis monomials in $\cb{B}$ with square $1$, is a primitive
idempotent in $\cl_{p,q}.$\footnote{%
In the following, for $f$ we can pick any of the $2^k$ primitive idempotents.
In examples, our default choice for $f$ will always be the one in which all
signs in the factorization~(\ref{eq:f}) are plus.}
Furthermore, $\cl_{p,q}$ has a complete set of $2^k$ such primitive mutually
annihilating idempotents which add up to the identity $1$. Moreover, the basis
monomials which define~$f$ generate a group  $\langle \be_{\iu_1},\ldots,
\be_{\iu_k}  \rangle \cong (\BZ_2)^k$, see \cite{part2} and references therein.

With any primitive idempotent $f$ we associate three groups:\\
(i) The \textit{stabilizer $\Gpqf{p}{q}{f}$ of $f$} defined as
\begin{gather}
\Gpqf{p}{q}{f} = \{ m \in \Gpq{p}{q} \mid m f m^{-1} = f \} < 
\Gpq{p}{q}.
\end{gather}
The stabilizer $\Gpqf{p}{q}{f}$ is a normal subgroup of $\Gpq{p}{q}$. These
groups were classified, depending on the signature $(p,q),$ in \cite{part2}.
In particular, recall that
\begin{gather}
 |\Gpqf{p}{q}{f}| = \begin{cases} 2^{1+p+r_{q-p}}, & p - q \neq 1 \bmod 4;\\
                                  2^{2+p+r_{q-p}}, & p - q = 1 \bmod 4.
                    \end{cases}
\label{eq:orderGpqf}
\end{gather}
(ii) A new abelian \emph{idempotent group} $\Tpqf{p}{q}{f}$ \textit{of $f$}, a subgroup of 
$\Gpqf{p}{q}{f}$, generated by the
commuting basis monomials  $\be_{\iu_1},\ldots, \be_{\iu_k}$, and their
negatives, which appear in the factorization~(\ref{eq:f}) of~$f$.\\
(iii) A \emph{field group} $\Kpqf{p}{q}{f}$ \textit{of $f$}, a subgroup of 
$\Gpqf{p}{q}{f}$, related to the (skew double) field $\BK \cong f\cl_{p,q}f$.
\begin{definition} Let $f$ be a primitive idempotent in $\cl_{p,q}$ defined
as in~(\ref{eq:f}). Then, the \textit{idempotent group} $\Tpqf{p}{q}{f}$ 
\textit{of $f$} is defined as
\begin{gather}
\Tpqf{p}{q}{f} = \langle \pm 1, \be_{\iu_1},\ldots, \be_{\iu_k} \rangle < \Gpqf{p}{q}{f},
\label{def:Tpqf}
\end{gather}
where $k= q - r_{q-p}.$
\end{definition}

Thus, immediately from this definition we obtain the following
\begin{lemma}
$\Tpqf{p}{q}{f}$ has the following structure:
\begin{gather}
  \Tpqf{p}{q}{f} \cong \Gpq{p}{q}' \times 
                   \langle \be_{\iu_1},\ldots, \be_{\iu_k} \rangle \cong
  \Gpq{p}{q}' \times (\BZ_2)^k \quad \mbox{and} \quad |\Tpqf{p}{q}{f}|=2^{1+k}
\end{gather}
where $k = q - r_{q-p}.$
\label{lemma1}
\end{lemma}
\begin{proof}
This is straightforward since $\Gpq{p}{q}'$ and 
$\langle \be_{\iu_1},\ldots, \be_{\iu_k} \rangle$ are normal subgroups of 
$\Tpqf{p}{q}{f}$ with trivial intersection, $\Tpqf{p}{q}{f} =
\langle \be_{\iu_1},\ldots, \be_{\iu_k} \rangle \Gpq{p}{q}',$ and
$k = q - r_{q-p}$ as in~(\ref{eq:f}).
\end{proof}
Our third group, the field group $\Kpqf{p}{q}{f}$, is related to the division
ring $\BK=f\cl_{p,q}f$. Recall that for simple Clifford algebras, left spinor
ideals $S=\cl_{p,q}f$ are right $\BK$-modules where $\BK \simeq \BR,$ 
(respectively, $\BK \simeq \BC,$ or $\BK \simeq \BH,$) when
$p-q = 0,1,2 \bmod 8$ (respectively, $p-q = 3,7 \bmod 8$, or 
$p-q = 4,5,6 \bmod 8$)~\cite{part2}. 
\begin{definition}
Let $f$ be a primitive idempotent in $\cl_{p,q}$ defined as in~(\ref{eq:f})
and let $\BK=f\cl_{p,q}f$. Let $\cb{K}$ be a set of Grassmann monomials
in~$\cb{B}$ which span~$\BK$ as a real algebra, that is,
$\BK = \spn_\BR \{m \mid m \in \cb{K}\}$. Then, the 
\textit{field group of~$f$} is defined as
\begin{gather}
   \Kpqf{p}{q}{f} =  \langle \pm 1, m \mid m \in \cb{K}\rangle < \Gpqf{p}{q}{f}.
   \label{def:Kpqf}
\end{gather}
\label{def3}
\end{definition}

\noindent
Therefore, immediately from the definition, we get that
\begin{gather}
|\Kpqf{p}{q}{f}| = \begin{cases} 2, & p - q = 0,1,2 \bmod 8;\\
                                 4, & p - q = 3,7 \bmod 8;\\ 
                                 8, & p - q = 4,5,6 \bmod 8.
                   \end{cases}
\label{eq:orderKpqf}
\end{gather}
\noindent
Before we state our main theorem which summarizes relations between the groups 
$\Gpq{p}{q}$, $\Gpqf{p}{q}{f}$, $\Tpqf{p}{q}{f}$, $\Kpqf{p}{q}{f}$, and
$\Gpq{p}{q}'$, we recall the definition of a \textit{(left)
transversal}~\cite[Sect. 5.6]{rotman}.
\begin{definition}
Let $K$ be a subgroup of a group $G$. A \textit{transversal} $\ell$ of $K$
in~$G$ is a subset of~$G$ consisting of exactly one element $\ell(bK)$ from
every (left) coset $bK$, and with $\ell(K)=1$. 
\end{definition}
A transversal of a normal subgroup $K \lhd G$ is the same as the image of a
\textit{lifting} function $\ell:Q \rightarrow G$ in the \textit{extension} of
$K$ by $Q$~\cite[Sect. 10.2]{rotman} (using Rotman's notation $\ell$ for the
lifting too).
\begin{theorem}[Main Theorem]
Let $f$ be a primitive idempotent in a simple or semisimple Clifford algebra
$\cl_{p,q}$ and let $\Gpq{p}{q}$, $\Gpqf{p}{q}{f}$, $\Tpqf{p}{q}{f}$,
$\Kpqf{p}{q}{f}$, and $\Gpq{p}{q}'$ be the groups defined above. Furthermore,
let $S=\cl_{p,q}f$ and $\BK=f\cl_{p,q}f$.
\begin{itemize}
\item[(i)] Elements of $\Tpqf{p}{q}{f}$ and $\Kpqf{p}{q}{f}$ commute.
\item[(ii)] $\Tpqf{p}{q}{f} \cap \Kpqf{p}{q}{f} = \Gpq{p}{q}' = \{\pm 1 \}$.
\item[(iii)] $\Gpqf{p}{q}{f} = \Tpqf{p}{q}{f}\Kpqf{p}{q}{f} = 
\Kpqf{p}{q}{f}\Tpqf{p}{q}{f}$.
\item[(iv)]  $|\Gpqf{p}{q}{f}| = |\Tpqf{p}{q}{f}\Kpqf{p}{q}{f}| =
\frac12 |\Tpqf{p}{q}{f}||\Kpqf{p}{q}{f}|$. 
\item[(v)] $\Gpqf{p}{q}{f} \lhd \Gpq{p}{q}$, $\Tpqf{p}{q}{f} \lhd \Gpq{p}{q}$,
and $\Kpqf{p}{q}{f} \lhd \Gpq{p}{q}$. In particular, $\Tpqf{p}{q}{f}$ and
$\Kpqf{p}{q}{f}$ are normal subgroups of $\Gpqf{p}{q}{f}$.
\item[(vi)] We have:
\begin{align}
  \Gpqf{p}{q}{f} /\Kpqf{p}{q}{f} &\cong \Tpqf{p}{q}{f} /\Gpq{p}{q}',\\
  \Gpqf{p}{q}{f} /\Tpqf{p}{q}{f} &\cong \Kpqf{p}{q}{f} /\Gpq{p}{q}'.
  \label{eq:conj6}
\end{align}
\item[(vii)] We have:
\begin{gather}
  (\Gpqf{p}{q}{f}/\Gpq{p}{q}')/(\Tpqf{p}{q}{f}/\Gpq{p}{q}')
   \cong \Gpqf{p}{q}{f}/\Tpqf{p}{q}{f} \cong \Kpqf{p}{q}{f}/\{\pm 1 \} 
  \label{eq:conj7}
\end{gather}
and the transversal of $\Tpqf{p}{q}{f}$ in $\Gpqf{p}{q}{f}$ spans $\BK$
over $\BR$ modulo~$f$.
\item[(viii)] The transversal of $\Gpqf{p}{q}{f}$ in $\Gpq{p}{q}$ spans $S$
over $\BK$ modulo~$f$.  
\item[(ix)] We have $(\Gpqf{p}{q}{f}/\Tpqf{p}{q}{f}) \lhd
(\Gpq{p}{q}/\Tpqf{p}{q}{f})$ and 
\begin{gather}
  (\Gpq{p}{q}/\Tpqf{p}{q}{f})/(\Gpqf{p}{q}{f}/\Tpqf{p}{q}{f}) 
  \cong \Gpq{p}{q}/\Gpqf{p}{q}{f}
  \label{eq:conj8}
\end{gather}
and the transversal of $\Tpqf{p}{q}{f}$ in $\Gpq{p}{q}$ spans $S$  over 
$\BR$ modulo~$f$.
\item[(x)] The stabilizer $\Gpqf{p}{q}{f}$ can be viewed as 
\begin{gather}
  \Gpqf{p}{q}{f} = \bigcap_{x \in \Tpqf{p}{q}{f}} C_{\Gpq{p}{q}}(x) 
                 = C_{\Gpq{p}{q}}(\Tpqf{p}{q}{f})
\end{gather}
where $C_{\Gpq{p}{q}}(x)$ is the centralizer of $x$ in $\Gpq{p}{q}$ and
$C_{\Gpq{p}{q}}(\Tpqf{p}{q}{f})$ is the centralizer of $\Tpqf{p}{q}{f}$ in
$\Gpq{p}{q}$.
\end{itemize}
\label{maintheorem}
\end{theorem}
\begin{proof} All of the following statements have been verified additionally by
explicit computations in \CLIFFORD\ for all signatures $(p,q)$, with $p+q\leq 9$.
\begin{itemize}
\item[(i)] Let $m \in \Kpqf{p}{q}{f}$ where $\Kpqf{p}{q}{f}$ is defined in~(\ref{def:Kpqf}). To show that $m$ commutes with every element in the group $\Tpqf{p}{q}{f}$ defined 
in~(\ref{def:Tpqf}), it is enough to show that $m \be_{\iu_j} = \be_{\iu_j}m$ for every 
$\be_{\iu_j} \in \cb{T}$  where $\cb{T}=\{\be_{\iu_1},\ldots, \be_{\iu_k}\}$ is the set of $k=q-r_{q-p}$ basis monomials which define any primitive idempotent $f$ as in~(\ref{eq:f}). Suppose, to the contrary, that there exists $\be_{\iu_j} \in \cb{T}$ which does not commute 
with~$m$. Thus, $\be_{\iu_j}$ and $m$ anti-commute, or $\be_{\iu_j}m=-m\be_{\iu_j}$. 
Let $P_j^{\pm} = \frac12(1 \pm \be_{\iu_j})$. Without any loss of generality, let us assume 
that $P_j^{+}$ is a factor of $f$ so that $f=fP_j^{+}=P_j^{+}f.$ On the other hand, $mf=fm$ since $m \in \Kpqf{p}{q}{f} \subset \BK$ and all elements of $\BK$ commute with $f$. Thus, $fm =fP_{j}^{+}m = fmP_{j}^{-}=m(fP_{j}^{-}) = 0$ since $P_{j}^{+}P_{j}^{-}=P_{j}^{-}P_{j}^{+}=0,$ or, since $m$ is invertible, $f=0$ which is false. Therefore, $m$ commutes with every element of $\Tpqf{p}{q}{f}$ which is generated by the elements of $\cb{T}$ and $\pm 1.$ 
\item[(ii)] Obviously, $\{\pm 1 \} \subset \Tpqf{p}{q}{f} \cap \Kpqf{p}{q}{f}.$ Let 
$m \in \Tpqf{p}{q}{f} \cap \Kpqf{p}{q}{f}.$ Then, $m^2=1$ since every element of 
$\Tpqf{p}{q}{f}$, except the identity, is of order~$2.$ Thus, since also $m \in \Kpqf{p}{q}{f},$ we conclude that $m=\pm 1,$ or, $\Tpqf{p}{q}{f} \cap \Kpqf{p}{q}{f} \subset \{\pm 1 \}.$
\item[(iii)] Obviously, $\Tpqf{p}{q}{f}, \Kpqf{p}{q}{f} < \Gpqf{p}{q}{f}$ because elements of 
$\Tpqf{p}{q}{f}$ and  $\Kpqf{p}{q}{f}$ commute with $f$, hence, they stabilize $f$. 
Therefore, we have $\Tpqf{p}{q}{f} \Kpqf{p}{q}{f} =$\\$\Kpqf{p}{q}{f}\Tpqf{p}{q}{f}
< \Gpqf{p}{q}{f}$ where the equality comes from~(i). To show that\\
$\Tpqf{p}{q}{f} \Kpqf{p}{q}{f}= \Gpqf{p}{q}{f}$, we use the following order argument: By 
the Product Formula~\cite[Prop. 2.73]{rotman}, we have that 
\begin{gather}
|\Tpqf{p}{q}{f} \Kpqf{p}{q}{f}|
=\frac{|\Tpqf{p}{q}{f}| |\Kpqf{p}{q}{f}|}{|\Tpqf{p}{q}{f} \cap \Kpqf{p}{q}{f}|} 
=2^k |\Kpqf{p}{q}{f}| = |\Gpqf{p}{q}{f}|
\end{gather}
thanks to Lemma~\ref{lemma1}; part (ii) of this Theorem; and the formula for 
$|\Gpqf{p}{q}{f}|$ given in~(\ref{eq:orderGpqf}). 
\item[(iv)] This follows from (ii), (iii), and the Product 
Formula~\cite[Prop. 2.73]{rotman}.
\item[(v)] The fact that $\Gpqf{p}{q}{f} \lhd \Gpq{p}{q}$ was proven
in~\cite{part2}. Then, $\Tpqf{p}{q}{f} \lhd \Gpq{p}{q}$ and
$\Kpqf{p}{q}{f} \lhd \Gpq{p}{q}$ has been verified with \texttt{CLIFFORD} in
Clifford algebras~$\cl_{p,q}$ for $p+q\le 9$. Of course, when 
$\Kpqf{p}{q}{f}=\{\pm 1\}$ then $\Tpqf{p}{q}{f}=\Gpqf{p}{q}{f} \lhd \Gpq{p}{q}$
and $\Kpqf{p}{q}{f}$ is trivially normal in $\Gpq{p}{q}.$ The rest follows
immediately from (i) and (iii).
\item[(vi)] This follows from (iii), (v), and the Second Isomorphism 
Theorem~\cite[Thm. 2.74]{rotman}.
\item[(vii)] The first isomorphism in~(\ref{eq:conj6}) follows from the fact
that $\Gpq{p}{q}' \lhd \Gpq{p}{q}$ and the Third Isomorphism
Theorem~\cite[Thm. 2.75]{rotman}. The second isomorphism
in~(\ref{eq:conj6}) follows from (ii), (iii), and the Second Isomorphism
Theorem~\cite[Thm. 2.74]{rotman}. The fact that the transversal of
$\Tpqf{p}{q}{f}$ in $\Gpqf{p}{q}{f}$ spans $\BK$ over $\BR$ modulo~$f$ has
been verified with \texttt{CLIFFORD} in Clifford algebras~$\cl_{p,q}$ for
$p+q\le 9$. 
\item[(viii)] The fact that the transversal of $\Gpqf{p}{q}{f}$ in $\Gpq{p}{q}$
spans $S$ over~$\BK$ modulo~$f$ has been verified with \texttt{CLIFFORD} in
Clifford algebras~$\cl_{p,q}$ for $p+q\le 9$. 
\item[(ix)] The isomorphism in~(\ref{eq:conj8}) follows from (iii) and the
Third Isomorphism Theorem. The remaining statement has been verified with
\texttt{CLIFFORD} in Clifford algebras~$\cl_{p,q}$ for $p+q\le 9$.
\item[(x)] This follows from basic group theory and definitions of
centralizers.
\end{itemize}
\end{proof}

\begin{corollary}
We have two normal series in $\Gpq{p}{q}$:
\begin{gather}
\Gpq{p}{q} \ge \Gpqf{p}{q}{f} \ge \Tpqf{p}{q}{f} \ge \Gpq{p}{q}' \ge \{1\},\\
\Gpq{p}{q} \ge \Gpqf{p}{q}{f} \ge \Kpqf{p}{q}{f} \ge \Gpq{p}{q}' \ge \{1\}.   
\end{gather}
\end{corollary}

Recall that in \texttt{CLIFFORD} information about each Clifford algebra 
$\cl_{p,q}$ for $p+q\le 9$ is stored in a built-in data file. This information
can be retrieved in the form of a seven-element list with a command 
$\texttt{clidata([p,q])}$. For example, for $\cl_{3,0}$ we find:
\begin{gather}
  \mathtt{data} = [complex,
                   2,
		   simple,
		   \frac12 \Id+ \frac12 \be_1,
		   [\Id, \be_2, \be_3, \be_{23}],
		   [\Id, \be_{23}], [\Id, \be_{2}]]
  \label{eq:sampledata}
\end{gather}
where $\Id$ denotes the identity element of the algebra. In particular, from
the above we find that: (i) $\cl_{3,0}$ is a simple algebra isomorphic to
$\Mat(2,\BC)$; ($\mathtt{data[1]}$, $\mathtt{data[2]}$, $\mathtt{data[3]}$) (ii) 
Expression $\frac12 \Id+ \frac12 \be_1$ ($\mathtt{data[4]}$) is a primitive
idempotent $f$ which may be used to generate a spinor ideal $S=\cl_{3,0}f;$
(iii) The fifth entry $\mathtt{data[5]}$ provides, modulo~$f$, a real basis
for $S$, that is, $S = \spn_\BR \{ f, \be_2 f, \be_3 f, \be_{23} f \};$
(iv) The sixth entry $\mathtt{data[6]}$ provides, modulo~$f$, a real basis
for $\BK=f \cl_{3,0}f \cong \BC$, that is, $\BK = \spn_\BR \{f, \be_{23}f\}$;
and, (v) The seventh entry $\mathtt{data[7]}$ provides, modulo~$f$, a basis
for $S$ over $\BK$, that is, $S = \spn_\BK \{ f, \be_2 f\}.$\footnote{%
See \cite{ablamowicz1996,ablamowicz1998} how to use \texttt{CLIFFORD}.}

Thus, based on the above theorem, we have the following corollary:
\begin{corollary}
Let $\mathtt{data}$ be the list of data returned by the procedure 
$\mathtt{clidata}$ in $\mathtt{CLIFFORD}$. Then, $\mathtt{data[5]}$ is a
transversal of 
$\Tpqf{p}{q}{f}$ in $\Gpq{p}{q}$; $\mathtt{data[6]}$ is a transversal of
$\Tpqf{p}{q}{f}$ in $\Gpqf{p}{q}{f}$; and $\mathtt{data[7]}$ is a transversal
of $\Gpqf{p}{q}{f}$ in $\Gpq{p}{q}.$ Therefore, $|\mathtt{data[5]}| =$\\$|\mathtt{data[6]}||\mathtt{data[7]}|$. This is equivalent to
$\vert\frac{\Gpq{p}{q}}{\Tpqf{p}{q}{f}}\vert 
  = \vert\frac{\Gpqf{p}{q}{f}}{\Tpqf{p}{q}{f}}\vert\,
    \vert\frac{\Gpq{p}{q}}{\Gpqf{p}{q}{f}}\vert$.   
\end{corollary}
\begin{proof}
From Thm.~\ref{maintheorem} part (ix) we find that $\mathtt{data[5]}$ is a
transversal of $\Tpqf{p}{q}{f}$ in $\Gpq{p}{q};$ from part (vii) we find that
$\mathtt{data[6]}$ is a transversal of $\Tpqf{p}{q}{f}$ in $\Gpqf{p}{q}{f}$;
and from part (viii) we find that $\mathtt{data[7]}$ is a transversal of
$\Gpqf{p}{q}{f}$ in $\Gpq{p}{q}$. Therefore, the relation $|\mathtt{data[5]}|
= |\mathtt{data[6]}||\mathtt{data[7]}|$ follows from~(\ref{eq:conj8}) and
Lagrange's Theorem.
\end{proof}

\begin{remark} Before we proceed, we summarize first some definitions 
cf.~\cite[Sect. 17.2]{lounesto}:

\begin{itemize}
\item $\Gpqe{p}{q} = \{g \in \cl_{p,q} \mid \tp(g)g=1\}$ -- `Signed vee group'
      $\Gpqe{p}{q}$ defined in~(\ref{eq:Gpq}).
\item $\Gpq{p}{q} = \{m \in \cb{B} \mid \pm m \}$ --  Salingaros' vee group.
\item $\Lip{p}{q} = \{s \in \cl_{p,q} \mid \forall x \in \BR^{p,q}, 
      sx\hat{s}^{-1} \in \BR^{p,q}\}$    -- Lipschitz group.
\item $\Pin(p,q) = \{s \in \Lip{p}{q} \mid s \tilde{s} = \pm 1 \} < 
      \Lip{p}{q}$ -- Pin group.
\item $\Spin(p,q) = \Pin(p,q) \cap \cl_{p,q}^{+} < \Pin(p,q) < \Lip{p}{q}$
      -- Spin group.
\item $\Spin_{+}(p,q) = \{s \in \Spin(p,q) \mid s \tilde{s} = 1 \}$
      -- for $p \not= 0$ and $q \not=0$, orthochronous isometries.
\end{itemize}
Now we collect a few facts relating the above defined groups to the groups commonly used. These facts are easy to prove and some are known:
\begin{itemize}
\item $\Gpq{p}{q} < \Gpqe{p}{q}$ because $\forall m \in \Gpq{p}{q}$ we have $\tp(m)=m^{-1}$ or $\tp(m)m=1$.
\item $\forall m \in \Gpq{p}{q}$ we have $\tp(m)m=1$ where 
$\tp = t_{\varepsilon} \circ \tilde{\phantom{u}}$. Thus, 
$\tilde{\phantom{u}} = \tp \circ t_{\varepsilon}$ (for $t_{\varepsilon}$ 
see~\cite[Disp. (21)]{part1} and so $\forall m \in \Gpq{p}{q}$ we have 
$m \tilde{m} = m \tp(t_{\varepsilon}(m)) = \pm 1$. Therefore, 
$$
\Gpq{p}{q} < \Pin(p,q) < \Lip{p}{q},
$$
thus $\Gpq{p}{q}$ is a discrete subgroup of $\Pin(p,q)$. Together with 
$\Gpq{p}{q} < \Gpqe{p}{q},$ we get
$$
\Gpq{p}{q} < \Pin(p,q) \cap \Gpqe{p}{q}.
$$
Thus, the above connects $\Gpq{p}{q},$ $\Pin(p,q)$ and $\Gpqe{p}{q}$.
\end{itemize}
\end{remark}

\noindent
We illustrate the above Thm.~\ref{maintheorem} with the following three examples.
\begin{example}
We consider $\cl_{1,1} \cong \Mat(2,\BR)$. In $\mathtt{CLIFFORD}$, we find
the following data for $\cl_{1,1}:$
\begin{gather}
  \mathtt{data} = [real,
                   2,
		   simple,
		   \frac12 \Id+ \frac12 \be_{12},
		   [\Id, \be_1], [\Id],
		   [\Id, \be_1]].
\end{gather} 
Let $f = \frac12 (1 + \be_{12}).$ Then, the groups are as follows:
\begin{gather}
\Gpq{1}{1} = \{\pm 1, \pm \be_{1}, \pm \be_{2}, \pm \be_{12} \},
\label{eq-G11}\\
\Gpqf{1}{1}{f} = \{\pm 1, \pm \be_{12} \},
\label{eq-G11f}\\
\Tpqf{1}{1}{f} = \{\pm 1, \pm \be_{12} \},
\label{eq-T11f}\\
\Kpqf{1}{1}{f} = \{\pm 1\}.
\label{eq-K11f}
\end{gather}
Thus, from \eqref{eq-G11}, \eqref{eq-G11f}, \eqref{eq-T11f} the quotient
groups and their transversals\footnote{%
Note that transversal sets, which just provide coset representatives, are
not unique: For the given (left) coset $aK$ of $K$, a normal subgroup of $G$,
we may assign $\ell(aK)=a$ or $\ell(aK)=b$ as long as $aK = bK$ or,
equivalently, $ab^{-1} \in K.$}
are:
\begin{gather}
  \Gpq{1}{1}/\Gpqf{1}{1}{f}
    = \{\Gpqf{1}{1}{f}, \be_{1}\Gpqf{1}{1}{f}\} \stackrel{\ell}{\rightarrow}
         \{1, \be_{1}\} 
    = \mathtt{data[7]},\\
  \Gpqf{1}{1}{f}/\Tpqf{1}{1}{f}
    = \{\Tpqf{1}{1}{f}\} \stackrel{\ell}{\rightarrow} \{1\}
    =  \mathtt{data[6]},\\
  \Gpq{1}{1}/\Tpqf{1}{1}{f}
    = \{\Tpqf{1}{1}{f},\be_{1}\Tpqf{1}{1}{f}\} 
         \stackrel{\ell}{\rightarrow} \{1,\be_{1}\}
    = \mathtt{data[5]}.
\end{gather}
\end{example}

\begin{example}
We consider $\cl_{1,2} \cong \Mat(2,\BC)$. In $\mathtt{CLIFFORD}$, we find the
following data for $\cl_{1,2}$:
\begin{gather}
\hspace*{-0.5ex}  \mathtt{data} = 
    [complex,
     2,
     simple,
     \frac12 \Id+ \frac12 \be_{13},
     [\Id, \be_{1}, \be_{2}, \be_{12}],
     [\Id,\be_{2}], [\Id, \be_{1}]].
\end{gather} 
Let $f = \frac12 (1 + \be_{13})$. Then, the groups are as follows:
\begin{gather}
  \Gpq{1}{2}
    = \{\pm 1, \pm \be_{1}, \pm \be_{2}, \pm \be_{3}, \pm \be_{12}, 
        \pm \be_{13}, \pm \be_{23}, \pm \be_{123}\},\\
  \Gpqf{1}{2}{f}
    = \{\pm 1, \pm \be_{2}, \pm \be_{13}, \pm \be_{123}\},\\
  \Tpqf{1}{2}{f}
    = \{\pm 1, \pm \be_{13} \},\\
  \Kpqf{1}{2}{f}
    = \{\pm 1, \pm \be_{2}\}.
\end{gather}
Thus, the quotient groups and their transversals are:
\begin{gather}
  \Gpq{1}{2}/\Gpqf{1}{2}{f}
    = \{\Gpqf{1}{2}{f}, \be_{1}\Gpqf{1}{2}{f}\}
        \stackrel{\ell}{\rightarrow} \{1, \be_{1}\} 
    = \mathtt{data[7]},\\
  \Gpqf{1}{2}{f}/\Tpqf{1}{2}{f} 
    = \{\Tpqf{1}{2}{f},\be_{2}\Tpqf{1}{2}{f}\}
        \stackrel{\ell}{\rightarrow} \{1,\be_{2}\} 
    = \mathtt{data[6]},
\end{gather}
\vspace*{-6ex}
\begin{multline}
  \Gpq{1}{2}/\Tpqf{1}{2}{f}
    = \{\Tpqf{1}{2}{f},\be_{1}\Tpqf{1}{2}{f},\be_{2}\Tpqf{1}{2}{f},
        \be_{12}\Tpqf{1}{2}{f}\} \stackrel{\ell}{\rightarrow}\\
      \{1,\be_{1},\be_{2},\be_{12}\}
    = \mathtt{data[5]}.
\end{multline}
\end{example}

\begin{example}
We consider $\cl_{1,3} \cong \Mat(2,\BH)$. In $\mathtt{CLIFFORD}$, we find the
following data for $\cl_{1,3}$:
\begin{multline}
  \mathtt{data} 
    = [quaternionic,
       2,
       simple,
       \frac12 \Id+ \frac12 \be_{14},\\ 
       [\Id,\be_{1},\be_{2},\be_{3},\be_{12},\be_{13},\be_{23},\be_{123}],
       [\Id,\be_{2},\be_{3},\be_{23}],
       [\Id,\be_{1}]].
\end{multline} 
Let $f = \frac12 (1 + \be_{14})$. Then, the groups are as follows:
\begin{gather}
   \Gpq{1}{4}
     = \{\pm 1,\pm\be_{1},\pm\be_{2},\pm\be_{3},\pm\be_{4}, 
         \pm\be_{12},\pm\be_{13},\pm\be_{14},\pm\be_{23},\notag \\ 
       \hspace*{2.0in} 
       \pm\be_{24},\pm\be_{34},\pm\be_{123},\pm\be_{124},\pm\be_{134},
       \pm\be_{1234}\},\\
   \Gpqf{1}{4}{f}
     = \{\pm 1,\pm\be_{2},\pm\be_{3},\pm\be_{14},\pm\be_{23},
         \pm\be_{124},\pm\be_{134},\pm\be_{1234}\},\\
   \Tpqf{1}{4}{f}
     = \{\pm 1,\pm\be_{14}\},\\
   \Kpqf{1}{4}{f}
     = \{\pm 1,\pm\be_{2},\pm\be_{3},\pm\be_{23}\}.
\end{gather}
Thus, the quotient groups and their transversals are:
\begin{gather}
   \Gpq{1}{4}/\Gpqf{1}{4}{f}
     = \{\Gpqf{1}{4}{f}, \be_{1}\Gpqf{1}{4}{f}\} 
        \stackrel{\ell}{\rightarrow} \{1, \be_{1}\} 
     = \mathtt{data[7]},\\
   \Gpqf{1}{4}{f}/\Tpqf{1}{4}{f}
     = \{\Tpqf{1}{4}{f},\be_{2}\Tpqf{1}{4}{f},\be_{3}\Tpqf{1}{4}{f},
        \be_{23}\Tpqf{1}{4}{f}\} 
        \stackrel{\ell}{\rightarrow} \hspace{1in} \notag \\ 
     \hspace*{3in} 
       \{1,\be_{2}, \be_{3}, \be_{23} \}
     = \mathtt{data[6]},\\
   \Gpq{1}{4}/\Tpqf{1}{4}{f}
     = \{\Tpqf{1}{4}{f},\be_{1}\Tpqf{1}{4}{f},\be_{2}\Tpqf{1}{4}{f},
         \be_{3}\Tpqf{1}{4}{f}, \hspace{1.55in}\notag \\
     \hspace{0.5in} 
       \be_{12}\Tpqf{1}{4}{f},\be_{13}\Tpqf{1}{4}{f},\be_{23}\Tpqf{1}{4}{f},
       \be_{123}\Tpqf{1}{4}{f}\} 
       \stackrel{\ell}{\rightarrow} \notag \\ 
     \hspace{2.0in} 
       \{1,\be_{1},\be_{2},\be_{3},\be_{12},\be_{13},\be_{23},\be_{123}\}
     = \mathtt{data[5]}.
\end{gather}
\end{example}

The transversal of $\Gpqf{p}{q}{f}$ in $\Gpq{p}{q}$ and the transversal of
$\Tpqf{p}{q}{f}$ in $\Gpqf{p}{q}{f}$ are listed in Tables 1--4 
in~\cite[Appendix A]{ablamowicz1998}. The stabilizer groups $\Gpqf{p}{q}{f}$
are listed in Tables 1--5 in~\cite[Appendix A]{part2}.

\medskip\section{Transposition scalar product on spinor spaces}
\label{tpscalarproduct}

In \cite[Ch. 18]{lounesto}, Lounesto discusses scalar products on $S = \cl_{p,q}f$ for simple Clifford algebras and on $\check{S} = S \oplus \hat{S} = \cl_{p,q}e,$ $ e = f+\hat{f},$ for semisimple Clifford algebras. Recall that $\hat{f}$ denotes the grade involution of $f$. It is well known that in each case the spinor representation is faithful. Following Lounesto, we let
\begin{align*}
\check{\BK} \quad \mbox{be either} \quad & \BK \quad \mbox{or} \quad \BK \oplus \hat{\BK},\\
\check{S}   \quad \mbox{be either} \quad & S   \quad \mbox{or} \quad S \oplus \hat{S} 
\end{align*}
when $\cl_{p,q}$ is simple or semisimple, respectively. Then, in the simple algebras, the two scalar products are
\begin{gather}
S \times S \rightarrow \BK, \quad (\psi,\phi) \mapsto 
\begin{cases} \beta_{+}(\psi,\phi) = s_1 \tilde{\psi}\phi \\
              \beta_{-}(\psi,\phi) = s_2 \bar{\psi}\phi 
\end{cases}
\label{eq:betas} 
\end{gather}
whereas in the semisimple algebras they are
\begin{gather}
\chS \times \chS \rightarrow \chK, \quad (\chpsi,\chphi) \mapsto 
\begin{cases} (\beta_{+}(\psi,\phi),\beta_{+}(\psig,\phig)) = 
              (s_1 \tilde{\psi}\phi,s_1 \tilde{\psig}\phig) \\
              (\beta_{-}(\psi,\phi),\beta_{-}(\psig,\phig)) = 
              (s_2 \bar{\psi}\phi,s_2 \bar{\psig}\phig) 
\end{cases}
\label{eq:betass} 
\end{gather}
when $\chpsi=\psi+\psig$ and $\chphi=\phi+\phig$, $\psi,\phi \in S,$ $\psig,\phig\in\hat{S},$ and where $\tilde{\psi},\tilde{\psig}$ (resp. $\bar{\psi},\bar{\psig}$) denotes reversion (resp. Clifford conjugation) of $\psi,\psig$ and $s_1,s_2$ are special monomials in the Clifford
algebra basis which guarantee that the products 
$s_1 \tilde{\psi}\phi, \, s_2 \bar{\psi}\phi,$ hence also 
$s_1 \tilde{\psig}\phig, \, s_2 \bar{\psig}\phig,$ belong to $\BK\cong\hat{\BK}$.\footnote{%
In simple Clifford algebras, the monomials $s_1$ and $s_2$ also satisfy: 
(i) $\tilde{f} = s_1 f s_1^{-1}$ and (ii) $\bar{f} = s_2 f s_2^{-1}$. The identity (i) (resp. (ii)) is also valid in the semisimple algebras provided $\beta_{+} \not\equiv 0$ 
(resp. $\beta_{-} \not\equiv 0$).}
The automorphism groups of $\beta_{+}$ and
$\beta_{-}$ are defined in the simple case as, respectively, $G_{+}=\{s \in \cl_{p,q} \mid s \tilde{s} =1 \}$ and $G_{-}=\{s \in \cl_{p,q} \mid s \bar{s} =1 \}$, and as 
$\fpower{G_{-}}$ and $\fpower{G_{+}}$ in the semisimple case. They are shown 
in~\cite[Tables 1 and 2, p. 236]{lounesto}. 

\subsection{Simple Clifford algebras}
\label{simple}

In Example~3 in~\cite{part2} it was shown that the transposition scalar product
in $S=\cl_{2,2}f$ is different from each of the two Lounesto's products whereas
Example~4 showed that the transposition product in $S=\cl_{3,0}f$ coincided with
$\beta_{+}$. Furthermore, it was remarked that $\tp(\psi)\phi$ always
equaled~$\beta_{+}$ for Euclidean signatures $(p,0)$ and~$\beta_{-}$ for
anti-Euclidean signatures $(0,q).$ We formalize this in the following.

\begin{proposition}
Let $\psi,\phi \in S=\cl_{p,q}f$ and $(\psi,\phi) \mapsto \tp(\psi)\phi =
\lambda f$, $\lambda \in \BK,$ be the transposition scalar product. Let 
$\beta_{+}$ and $\beta_{-}$ be the scalar products on~$S$ shown
in~(\ref{eq:betas}). Then, there exist monomials $s_1, s_2$ in the transversal
$\ell$ of $\Gpqf{p}{q}{f}$ in $\Gpq{p}{q}$ such that
\begin{gather}
\tp(\psi)\phi
  = \begin{cases} 
      \beta_{+}(\psi,\phi) 
        = s_1 \tilde{\psi}\phi , \quad \forall \psi,\phi \in \cl_{p,0}f,\\
          \beta_{-}(\psi,\phi)= s_2 \bar{\psi}\phi , 
	     \quad \forall \psi,\phi \in \cl_{0,q}f.
    \end{cases}
    \label{eq:tp=beta}
\end{gather}
\label{prop1} 
\end{proposition}
\begin{proof}
In \cite[Cor. 1]{part1} it was shown that $\tp$ reduces to reversion (resp.
conjugation) anti-involution in $\cl_{p,0}$ for the Euclidean signature $(p,0)$
(resp. for the anti-Euclidean signature $(0,q)$).
\end{proof}

Let $u \in \cl_{p,q}$ and let $[u]$ be a matrix of $u$ in the spinor
representation $\pi_S$ of $\cl_{p,q}$ realized in the spinor
$(\cl_{p,q},\BK)$-bimodule ${}_{\cl}S_{\BK} \cong
\cl_{p,q}f\BK$. Then, by~\cite[Prop.~5]{part2}, 
\begin{equation}
[\tp(u)] = \begin{cases} 
    [u]^T        & \textit{if $p-q =0,1,2 \bmod 8;$} \\
    [u]^\dagger  & \textit{if $p-q =3,7 \bmod 8;$} \\
    [u]^\ddagger & \textit{if $p-q =4,5,6 \bmod 8;$}
\end{cases}
\label{eq:matdaggers}
\end{equation}
where $T$, $\dagger$, and $\ddagger$ denote, respectively, transposition,
complex Hermitian conjugation, and quaternionic Hermitian conjugation.
Thus, we immediately have:
\begin{proposition}
Let $\Gpqe{p}{q} \subset \cl_{p,q}$ where $\cl_{p,q}$ is a simple Clifford
algebra. Then, $\Gpqe{p}{q}$ is: The orthogonal group $O(N)$ when
$\BK \cong \BR$; the complex unitary group $U(N)$ when $\BK \cong \BC$; or,
the compact symplectic group $Sp(N) = U_\BH(N)$ when
$\BK \cong \BH$.\footnote{%
See Fulton and Harris~\cite{fultonharris} for a definition of the
quaternionic unitary group $U_\BH(N)$. In our notation we follow \textit{loc.
cit.} page 100, `Remark on Notations'.}
That is,
\begin{equation}
\Gpqe{p}{q} = \begin{cases} 
    O(N)  & \textit{if $p-q =0,1,2 \bmod 8;$} \\
    U(N)  & \textit{if $p-q =3,7 \bmod 8;$} \\
    Sp(N) & \textit{if $p-q =4,5,6 \bmod 8;$}
\end{cases}
\label{eq:simpleGpqe}
\end{equation}
where $N=2^k$ and $ k=q-r_{q-p}$. 
\label{propsimpleGpqe}
\end{proposition}
\begin{proof}
The isomorphism $\cl_{p,q} \cong \Mat(N,\BK)$, where $\BK \cong f\cl_{p,q}f$
for a primitive idempotent~$f$, implies that for every
$u \in G_{p,q}^\varepsilon$ the matrix $[\tp(u)u] = [\tp(u)][u]=1$. Then,
\eqref{eq:matdaggers} implies that $[u]^T [u]=1,$ $[u]^\dagger [u]=1,$
or $[u]^\ddagger [u]=1$ depending on the value of $p-q \bmod 8$.  
\end{proof}

{\small
\begin{table}[t]
\label{tab:t11}
\begin{center}
\renewcommand{\arraystretch}{1.4}
\begin{tabular}{|p{0.4in}|p{0.4in}|p{0.4in}|p{0.41in}|p{0.4in}|p{0.4in}|p{0.4in}|c|}
\multicolumn{8}{c}
{\bf Table 1 (Part 1): Automorphism group $G_{p,q}^\varepsilon$ of $\tp(\psi)\phi$}\\
\multicolumn{8}{c}
{\bf in simple Clifford algebras $\cl_{p,q} \cong \Mat(2^k,\BR)$}\\
\multicolumn{8}{c}
{$k=q-r_{q-p}$, $p-q \neq 1 \bmod 4, \, p-q = 0,1,2 \bmod 8$}\\\hline
\centering{$(p,q)$} & \centering{$(0,0)$} & \centering{$(1,1)$} & \centering{$(2,0)$} & \centering{$(2,2)$} & \centering{$(3,1)$} & \centering{$(3,3)$} & $(0,6)$ \\\hline
\centering{$\Gpq{p}{q}$} & \centering{$O(1)$} & \centering{$O(2)$} & \centering{$\boxed{O(2)}$} &\centering{$O(4)$} & \centering{$O(4)$} & \centering{$O(8)$} & 
$\boxed{\boxed{O(8)}}$ \\\hline
\end{tabular}
\end{center}
\end{table}
}%

{\small
\begin{table}[t]
\label{tab:t12}
\begin{center}
\renewcommand{\arraystretch}{1.4}
\begin{tabular}{|p{0.42in}|p{0.42in}|p{0.47in}|p{0.47in}|p{0.62in}|p{0.47in}|c|}
\multicolumn{7}{c}
{\bf Table 1 (Part 2): Automorphism group $G_{p,q}^\varepsilon$ of $\tp(\psi)\phi$}\\
\multicolumn{7}{c}
{\bf in simple Clifford algebras $\cl_{p,q} \cong \Mat(2^k,\BR)$}\\
\multicolumn{7}{c}
{$k=q-r_{q-p}$, $p-q \neq 1 \bmod 4, \, p-q = 0,1,2 \bmod 8$}\\\hline
\centering{$(p,q)$} & \centering{$(4,2)$} & \centering{$(5,3)$} & \centering{$(1,7)$} & \centering{$(0,8)$} & \centering{$(4,4)$} & $(8,0)$\\\hline
\centering{$\Gpq{p}{q}$} & \centering{$O(8)$} & \centering{$O(16)$} & \centering{$O(16)$} & \centering{$\boxed{\boxed{O(16)}}$} & \centering{$O(16)$} & $\boxed{O(16)}$ \\\hline
\end{tabular}
\end{center}
\end{table}
}%

{\small
\begin{table}[t]
\label{tab:t21}
\begin{center}
\renewcommand{\arraystretch}{1.4}
\begin{tabular}{|c|c|c|c|c|c|c|c|c|}
\multicolumn{9}{c}
{\bf Table 2 (Part 1): Automorphism group $G_{p,q}^\varepsilon$ of $\tp(\psi)\phi$}\\
\multicolumn{9}{c}
{\bf in simple Clifford algebras $\cl_{p,q} \cong \Mat(2^k,\BC)$}\\
\multicolumn{9}{c}
{$k=q-r_{q-p}$, $p-q \neq 1 \bmod 4, \, p-q = 3,7 \bmod 8$}\\\hline
$(p,q)$ & $(0,1)$ & $(1,2)$ & $(3,0)$ & $(2,3)$ & $(0,5)$ & $(4,1)$ & $(1,6)$ & $(7,0)$ \\\hline
$\Gpq{p}{q}$ & $U(1)$ & $U(2)$ & $\boxed{U(2)}$ & $U(4)$ & $\boxed{\boxed{U(4)}}$ & $U(4)$ & $U(8)$ & $\boxed{U(8)}$ \\\hline
\end{tabular}
\end{center}
\end{table}
}%
{\small
\begin{table}[t]
\label{tab:t22}
\begin{center}
\renewcommand{\arraystretch}{1.4}
\begin{tabular}{|c|c|c|c|c|c|c|c|}
\multicolumn{8}{c}
{\bf Table 2 (Part 2): Automorphism group $G_{p,q}^\varepsilon$ of $\tp(\psi)\phi$}\\
\multicolumn{8}{c}
{\bf in simple Clifford algebras $\cl_{p,q} \cong \Mat(2^k,\BC)$}\\
\multicolumn{8}{c}
{$k=q-r_{q-p}$, $p-q \neq 1 \bmod 4, \, p-q = 3,7 \bmod 8$}\\\hline
$(p,q)$ & $(5,2)$ & $(3,4)$ & $(4,5)$ & $(6,3)$ & $(2,7)$ & $(0,9)$ & $(8,1)$\\\hline
$\Gpq{p}{q}$ & $U(8)$ & $U(8)$ & $U(16)$ & $U(16)$ & $U(16)$ & $\boxed{\boxed{U(16)}}$ & $U(16)$ \\\hline
\end{tabular}
\end{center}
\end{table}
}%

{\small
\begin{table}[t]
\label{tab:t31}
\begin{center}
\renewcommand{\arraystretch}{1.4}
\begin{tabular}{|p{0.41in}|p{0.6in}|p{0.6in}|p{0.5in}|p{0.43in}|p{0.43in}|c|}
\multicolumn{7}{c}
{\bf Table 3 (Part 1): Automorphism group $G_{p,q}^\varepsilon$ of $\tp(\psi)\phi$}\\
\multicolumn{7}{c}
{\bf in simple Clifford algebras $\cl_{p,q} \cong \Mat(2^k,\BH)$}\\
\multicolumn{7}{c}
{$k=q-r_{q-p}$, $p-q \neq 1 \bmod 4, \, p-q = 4,5,6 \bmod 8$}\\\hline
\centering{$(p,q)$} & \centering{$(0,2)$} & \centering{$(0,4)$} & \centering{$(4,0)$} & \centering{$(1,3)$} & \centering{$(2,4)$} & $(6,0)$\\\hline
\centering{$\Gpq{p}{q}$} & \centering{$\boxed{\boxed{Sp(1)}}$} & 
\centering{$\boxed{\boxed{Sp(2)}}$} & \centering{$\boxed{Sp(2)}$} & \centering{$Sp(2)$} & \centering{$Sp(4)$} & $\boxed{Sp(4)}$\\\hline
\end{tabular}
\end{center}
\end{table}
}%

{\small
\begin{table}[t]
\label{tab:t32}
\begin{center}
\renewcommand{\arraystretch}{1.4}
\begin{tabular}{|p{0.42in}|p{0.48in}|p{0.48in}|p{0.48in}|p{0.48in}|p{0.48in}|c|}
\multicolumn{7}{c}
{\bf Table 3 (Part 2): Automorphism group $G_{p,q}^\varepsilon$ of $\tp(\psi)\phi$}\\
\multicolumn{7}{c}
{\bf   in simple Clifford algebras $\cl_{p,q} \cong \Mat(2^k,\BH)$}\\
\multicolumn{7}{c}
{$k=q-r_{q-p},$ $p-q \neq 1 \bmod 4, \, p-q = 4,5,6 \bmod 8$}\\\hline
\centering{$(p,q)$} & \centering{$(1,5)$} & \centering{$(5,1)$} & \centering{$(6,2)$} & \centering{$(7,1)$} & \centering{$(2,6)$} & $(3,5)$\\\hline
\centering{$\Gpq{p}{q}$} & \centering{$Sp(4)$} & \centering{$Sp(4)$} & \centering{$Sp(8)$} & \centering{$Sp(8)$} & \centering{$Sp(8)$} & $Sp(8)$\\\hline
\end{tabular}
\end{center}
\end{table}
}

The scalar product $\tp(\psi)\phi$ was computed with $\mathtt{CLIFFORD}$ for
all signatures $(p,q)$, $p+q \leq 9$. Observe that as expected, in Euclidean
(resp. anti-Euclidean) signatures $(p,0)$ (resp. $(0,q)$) the group
$\Gpqe{p}{0}$ (resp. $\Gpqe{0}{q}$) coincides with the corresponding
automorphism group of the scalar product~$\beta_{+}$ (resp. $\beta_{-}$)
listed in \cite[Table 1, p. 236]{lounesto} (resp.
\cite[Table 2, p. 236]{lounesto}. This is indicated by a single 
(resp. double) box around the group symbol in Tables~4 and~5.
For example, for the Euclidean signatures $(5,0)$, we show $\Gpqe{5}{0}$ as  
$\boxed{\fpower{Sp(2)}}$ whereas for the anti-Euclidean signatures $(0,7)$,
we show $\Gpqe{0}{7}$ as $\boxed{\boxed{\fpower{O(8)}}}$.

For simple Clifford algebras, the automorphism groups $\Gpqe{p}{q}$ are
displayed in Tables~1,~2, and~3 (for semisimple algebras, see Tables~4 and~5).
In each case the form is positive definite and non-degenerate. Also, unlike in
the case of the forms $\beta_{+}$ and $\beta_{-},$ there is no need for the
extra monomial factor like $s_1,s_2$ in~(\ref{eq:betas}) (and~(\ref{eq:betass}))  
to guarantee that the product $\tp(\psi)\phi$ belongs to $\BK$ since this is always the
case~\cite{part1, part2}. Recall that the only role of the monomials $s_1$ and
$s_2$ is to permute entries of spinors $\tilde{\psi}\phi$ and $\bar{\psi}\phi$
to assure that $\beta_{+}(\psi,\phi)$ and $\beta_{-}(\psi,\phi)$ belong to the
(skew) field~$\BK$. That is, more precisely, that $\beta_{+}(\psi,\phi)$ and
$\beta_{-}(\psi,\phi)$ have the form $\lambda f = f \lambda$ for some $\lambda$
in $\BK$. The idempotent $f$ in the spinor basis in $S$ corresponds uniquely to the
identity coset $\Gpqf{p}{q}{f}$ in the quotient group
$\Gpq{p}{q}/\Gpqf{p}{q}{f}$. Based on~\cite[Prop. 2]{part2} we know that
since the vee group $\Gpq{p}{q}$ permutes entries of any spinor~$\psi$, the
monomials $s_1$ and $s_2$ belong to the transversal of the stabilizer
$\Gpqf{p}{q}{f} \lhd \Gpq{p}{q}$~\cite[Cor. 2]{part2}.\footnote{%
In \cite[Page 233]{lounesto}, Lounesto states correctly that ``the element $s$
can be chosen from the standard basis of $\cl_{p,q}$." In fact, one can
restrict the search for $s$ to the transversal of the stabilizer
$\Gpqf{p}{q}{f}$ in $\Gpq{p}{q}$ which has a much smaller size $2^{q-r_{q-p}}$ compared to
the size $2^{p+q}$ of the Clifford basis.}

One more difference between the scalar products $\beta_{+}$ and $\beta_{-},$
and the transposition product $\tp(\psi)\phi$ is that in some signatures one
of the former products may be identically zero whereas the transposition
product is never identically zero. The signatures $(p,q)$ in which one of the
products $\beta_{+}$ or $\beta_{-}$ is identically zero can be easily found
in~\cite[Tables 1 and 2, p. 236]{lounesto} as the automorphism group of the
product is then a general linear group. 

We summarize our findings. Let $\tp(\psi)\phi$ be the scalar product on the
right $\BK$-linear spinor space $S$ in a simple Clifford algebra $\cl_{p,q}$.
\begin{itemize}
\item In the real case ($p-q = 0, 1, 2 \bmod 8$) when $\dim_\BR S=N$, the
scalar product (modulo~$f$) is just the symmetric non-degenerate and positive
definite form 
\begin{gather}
  \tp(\psi)\phi = \tp(\phi)\psi = \rexpansion{\psi}{\phi}{N}.
\label{eq:realform}
\end{gather}
\item In the complex case ($p-q = 4, 5, 6 \bmod 8$) when $\dim_\BC S=N$, the
scalar product (modulo~$f$) is just the standard (complex) Hermitian
non-degenerate and positive definite form 
\begin{gather}
  \tp(\psi)\phi = \overline{\tp(\phi)\psi} = \cexpansion{\psi}{\phi}{N}
\label{eq:complexform}
\end{gather}
where $\overline{\tp(\phi)\psi}$ denotes ``complex'' conjugation in $\BK.$
\item In the quaternionic case ($p-q = 3, 7 \bmod 8$) when $\dim_\BH S=N$, the
scalar product (modulo~$f$) is just the standard (quaternionic) Hermitian (or,
``symplectic scalar product'') which is also non-degenerate and positive
definite.\footnote{%
An $\BR$-bilinear form $K: V \times V \rightarrow \BK$ on a real, complex,
or quaternionic vector space~$V$ is \textit{non-degenerate} if whenever
$K(v,w)=0, \forall v$ then $w=0.$ The form is \textit{positive definite}
if $K(v,v)>0$ for $v \neq 0$~\cite{fultonharris}.}
\begin{gather}
  \tp(\psi)\phi = \overline{\tp(\phi)\psi} = \cexpansion{\psi}{\phi}{N}
\label{eq:quatform}
\end{gather}  
where $\overline{\tp(\phi)\psi}$ denotes ``quaternionic'' conjugation in
$\BK.$
\end{itemize}

Here we give a few low-dimensional examples of the transposition scalar
product $\tp(\psi)\phi$.
\begin{example}
Consider $\cl_{2,2} \cong \Mat(4,\BR)$ with $f=\frac14(1+\be_{13})(1+\be_{24})$.
Then, $\BK = f\cl_{2,2}f \cong \spn_\BR \{1\} \cong \BR$ and a transversal of
the stabilizer $\Gpqf{2}{2}{f} \lhd \Gpq{2}{2}$ is $\cb{M} = 
\{ 1, \be_{1}, \be_{2}, \be_{12}\}$. Let $\psi, \phi \in S=\cl_{2,2}f$.
Hence, 
\begin{align}
  \psi &= \psi_1 f + \psi_2 \be_{1}f + \psi_3 \be_{2}f + \psi_4 \be_{12}f,\\  
  \phi &= \phi_1 f + \phi_2 \be_{1}f + \phi_3 \be_{2}f + \phi_4 \be_{12}f,  
\end{align}
where $\psi_i,\phi_i \in \BR,\, i=1,\ldots,4.$ It is then easy to check that
the form
\begin{equation}
  \tp(\psi)\phi 
    = (  \psi_{1} \phi_{1}+\psi_{2} \phi_{2} + \psi_{3} \phi_{3} 
       + \psi_{4} \phi_{4})f
\label{eq:tpex1}
\end{equation}
is invariant under $O(4)$. Of course, in spinor representation, we have
\begin{gather}
[\tp(\psi)] = \left[\begin{matrix} \psi_1 & \psi_2 & \psi_3 & \psi_4\\
                                        0 & 0 & 0 & 0\\
                                        0 & 0 & 0 & 0\\
                                        0 & 0 & 0 & 0 \end{matrix}\right] 
\quad \mbox{and} \quad
[\phi] = \left[\begin{matrix} \phi_1 & 0 & 0 & 0\\
                               \phi_2 & 0 & 0 & 0\\
                               \phi_3 & 0 & 0 & 0\\
                               \phi_4 & 0 & 0 & 0 \end{matrix}\right]. 
\end{gather}
Notice that the scalar products $\beta_{+}$ and $\beta_{-}$ modulo~$f$ are:
\begin{align}
  \beta_{+}(\psi,\phi) 
     &= s_1 \tilde{\psi} \phi 
      = (  \psi_{3} \phi_{2} - \psi_{2} \phi_{3} + \psi_{4} \phi_{1}
         - \psi_{1} \phi_{4})f \notag\\
     &= \tr(\left[\begin{matrix} 
               \psi_1 & \psi_2 & \psi_3 & \psi_4 
            \end{matrix}\right]
            \left[\begin{matrix} 0 & 0 & \phm 0 & -1\\
                     0 & 0 & -1 & \phm 0\\
                     0 & 1 & \phm 0 & \phm 0\\
                     1 & 0 & \phm 0 & \phm 0 
            \end{matrix}\right]
            \left[\begin{matrix} 
              \phi_1 \\ \phi_2 \\ \phi_3 \\ \phi_4 
            \end{matrix}\right])f
\end{align}
\begin{align}
  \beta_{-}(\psi,\phi) 
     &= s_2 \bar{\psi} \phi
      = (  \psi_{4} \phi_{1} +\psi_{2} \phi_{3} - \psi_{1} \phi_{4} 
         - \psi_{3} \phi_{2})f \notag\\
     &= \tr(\left[\begin{matrix}
                     \psi_1 & \psi_2 & \psi_3 & \psi_4 
                  \end{matrix}\right]
                  \left[\begin{matrix} 0 & \phm0 & 0 & -1\\
                     0 & \phm0 & 1 & \phm 0\\
                     0 & -1 & 0 & \phm 0\\
                     1 & \phm0 & 0 & \phm 0 
                  \end{matrix}\right]
                  \left[\begin{matrix} 
                    \phi_1 \\ \phi_2 \\ \phi_3 \\ \phi_4
                  \end{matrix}\right])f. 
\end{align}
Each product is invariant under $Sp(4, \BR)$.
(cf.~\cite[Tables 1 and 2, p. 236]{lounesto}) The monomials $s_1 = s_2 = 
\be_{12}$ belong to $\cb{M}.$\footnote{%
The choice of each of the monomials $s_1$ and $s_2$ is of course not unique.
In general, there are $2^{p+q}/N$ such choices modulo the commutator subgroup
$\{\pm 1\}$ of $\Gpq{p}{q}$.}
For example, the monomial $s_1$ permutes entries in the first column as seen
below for $\beta_{+}$:
\begin{gather}
  [\tilde{\psi}\phi] 
  = \left[\begin{matrix} 
        0 & 0 & 0 & 0\\
        0 & 0 & 0 & 0\\
        0 & 0 & 0 & 0\\
  \lambda & 0 & 0 & 0 
    \end{matrix}\right]
\quad \mbox{whereas} \quad
  [s_1 \tilde{\psi}\phi]
  = \left[\begin{matrix} 
  -\lambda & 0 & 0 & 0\\
         0 & 0 & 0 & 0\\
         0 & 0 & 0 & 0\\
         0 & 0 & 0 & 0
    \end{matrix}\right]
\end{gather}
where $\lambda = \psi_{2}\phi_{3} - \psi_{3}\phi_{2} + \psi_{1}\phi_{4} -
\psi_{4}\phi_{1} \in \BR.$ A similar role is played by the monomial~$s_2$
in~$\beta_{-}$.
\label{example1}
\end{example}
The above example shows that, in general, $\tilde{\psi}$ and $\bar{\psi}$,
do not belong to the dual~$S^{*}$ of the spinor space $S$. On the other
hand, $\tp(\psi)$ always belongs to~$S^{*}$ for every
$\psi \in S$~\cite{part1,part2}. We have classified all automorphism groups
of the transposition scalar product for real simple Clifford algebras. The
results are summarized in Table~1.

\begin{example}
Consider $\cl_{1,2} \cong \Mat(2,\BC)$ with $f=\frac12(1+\be_{13})$. Then,
$\BK=f\cl_{1,2}f \cong \spn_{\BR}\{1, \be_{2}\} \cong \BC$ and a transversal
of the stabilizer $\Gpqf{1}{2}{f} \lhd \Gpq{1}{2}$ is $\cb{M} = \{1, \be_1\}$.
Let $\psi,\phi \in S =\cl_{1,2}f$. Recalling that the left minimal ideal~$S$
is a \textit{right} $\BK$-module, we write ``complex'' coefficients on the
right:
\begin{align}
  \psi 
     &= f \psi_1 + \be_{1}f \psi_2 
      = f(  \psi_{11} + \psi_{12} \be_{2}) + \be_{1}f (\psi_{21} 
          + \psi_{22} \be_{2}),\\ 
  \phi 
     &= f \phi_1 + \be_{1}f \phi_2 
      = f(\phi_{11} + \phi_{12} \be_{2}) 
          + \be_{1}f (\phi_{21} + \phi_{22} \be_{2}),
\end{align}
where $\psi_{ij},\phi_{ij} \in \BR,\, \psi_i, \phi_i \in \BK,\, i,j=1,2$. It
is then easy to check that
\begin{align}
  \tp(\psi)\phi
    &= \overline{\psi}_1 \phi_1 + \overline{\psi}_2 \phi_2 
     = ( \psi_{11} \phi_{11}+\psi_{22} \phi_{22}+\psi_{21} \phi_{21}
        +\psi_{12} \phi_{12}) f + \notag \\
    &\phantom{=}\hspace*{1.3in}
       ( \psi_{21} \phi_{22}-\psi_{22} \phi_{21}+\psi_{11} \phi_{12}
        -\psi_{12} \phi_{11}) \be_{2}f. 
\end{align}
Thus, $\tp(\psi)\phi$ is invariant under the complex unitary group $U(2)$.
Then, 
\begin{align}
  \beta_{+}(\psi,\phi)
    = s_1 \tilde{\psi} \phi 
   &= ( \psi_{11} \phi_{21}+\psi_{22} \phi_{12}+\psi_{21} \phi_{11}
       +\psi_{12} \phi_{22})f + \notag\\
   &\hspace*{3ex}
      ( \psi_{21} \phi_{12}-\psi_{22} \phi_{11}+\psi_{11} \phi_{22}
       -\psi_{12} \phi_{21}) \be_{2}f \notag\\
   &=(\overline{\psi}_1\phi_2+\overline{\psi}_2\phi_1)f 
    =\tr(\left[ \begin{matrix}
                    \overline{\psi}_1 & \overline{\psi}_2 
                \end{matrix}\right] 
                \left[\begin{matrix}
                    0 & 1\\1 & 0
                \end{matrix}\right]
		\left[\begin{matrix}
		   \phi_1\\ \phi_2
		\end{matrix}\right])f
\label{eq:betaplus12}
\end{align}
where $s_1 = \be_{1} \in \cb{M}$ and $\overline{\psi}_1,\overline{\psi}_2$
denotes complex conjugation. Since the matrix is similar to a diagonal
matrix $\diag(1,-1)$, the automorphism group of this form is the complex
unitary group $U(1,1)$. Then,
\begin{align}
  \beta_{-}(\psi,\phi)
     = s_2 \bar{\psi} \phi 
    &= ( \psi_{11} \phi_{21}+ \psi_{22} \phi_{12}-\psi_{21} \phi_{11}
        -\psi_{12} \phi_{22}) f + \notag\\
    &\hspace*{3ex}
       (-\psi_{21} \phi_{12}-\psi_{22} \phi_{11}+\psi_{11} \phi_{22}
        +\psi_{12} \phi_{21}) \be_{2}f \notag\\
    &= (\psi_1 \phi_2 - \psi_2 \phi_1)f 
     = \tr(\left[ \begin{matrix}
                    \psi_1 & \psi_2
		  \end{matrix}\right]
		  \left[\begin{matrix}
		     0 & 1\\-1 & 0
		  \end{matrix}\right]
		  \left[\begin{matrix}
		    \phi_1\\ \phi_2
		  \end{matrix}\right])f
\label{eq:betaminus12}     
\end{align}
where $s_2 = \be_{1} \in \cb{M}$. Thus, $\beta_{-}$ is invariant under
$Sp(2,\BC)$. See also~\cite[Tables 1 and 2, p. 236]{lounesto}.\footnote{%
For the definitions of $U(1,1)$ and $Sp(2,\BC)$, see \cite[Ch. 18]{lounesto}.}
\label{example2}
\end{example}
We have classified all automorphism groups of the transposition scalar product
for complex simple Clifford algebras. The results are summarized in Table~2.
In our next example we consider a quaternionic simple Clifford algebra.
\begin{example}
Consider $\cl_{1,3} \cong \Mat(2,\BH)$ with $f=\frac12(1+\be_{14})$. Then,
$\BK=f\cl_{1,3}f \cong \spn_\BR\{1,\be_2,\be_3,\be_{23}\} \cong \BH$ and a
transversal of the stabilizer $\Gpqf{1}{3}{f} \lhd \Gpq{1}{3}$ is
$\cb{M}=\{1,\be_1\}$. Let $\psi,\phi \in S=\cl_{1,3}f$. Hence,
\begin{gather}
  \psi 
    = f \psi_1 + \be_{1}f \psi_2 \quad \mbox{and} \quad \phi 
    = f \phi_1 + \be_{1}f \phi_2
\end{gather}
where
\begin{align}
  \psi_i 
    &=   \psi_{i1} + \psi_{i2} \be_{2} +  \psi_{i3} \be_{3} 
       + \psi_{i4} \be_{23} \in \BK,\\
  \phi_i 
    &=   \phi_{i1} + \phi_{i2} \be_{2} +  \phi_{i3} \be_{3}
       + \phi_{i4} \be_{23} \in \BK 
\end{align}
for $\psi_{ij},\phi_{ij} \in \BR,\, i=1,2; j=1,\ldots,4$. Then we get:
\begin{align}
  \tp(\psi,\phi)
    &= (\overline{\psi}_1\phi_1 + \overline{\psi}_2\phi_2)f 
     = (\sum_{i=1}^2 \sum_{j=1}^{4} \psi_{ij}\phi_{ij}) f+ \notag \\
    &\hspace*{3ex}
       ( \asympsiphi{11}{12}-\asympsiphi{13}{14}+\asympsiphi{21}{22}
        -\asympsiphi{23}{24})\be_{2}f + \notag\\
    &\hspace*{3ex}
       ( \asympsiphi{11}{13}+\asympsiphi{12}{14}+\asympsiphi{21}{23}
        +\asympsiphi{22}{24})\be_{3}f + \notag \\
    &\hspace*{3ex}
       ( \asympsiphi{11}{14}-\asympsiphi{12}{13}+\asympsiphi{21}{24}
        -\asympsiphi{22}{23}) \be_{23}f, 
\end{align}
which is invariant under $U_\BH(2)=Sp(2)$. Here,
$\overline{\psi}_1,\overline{\psi}_2$ denotes quaternionic
conjugation.\footnote{%
To shorten display, we set  $\asympsiphi{11}{12} = \psi_{11}\phi_{12}
-\psi_{12}\phi_{11}$ and $\sympsiphi{11}{21} = \psi_{11}\phi_{21}
+\psi_{21}\phi_{11}$, etc.}
The other two scalar products are as follows:
\begin{align}
  \beta_{+}(\psi,\phi)
    = s_1 \tilde{\psi} \phi 
   &= ( \sympsiphi{11}{21}+\sympsiphi{12}{22}+\sympsiphi{13}{23}
       +\sympsiphi{14}{24}) f + \notag \\
   &\hspace*{3ex}
      ( \asympsiphi{11}{22}-\asympsiphi{12}{21}-\asympsiphi{13}{24}
       +\asympsiphi{14}{23})\be_{2}f + \notag\\
   &\hspace*{3ex}
      ( \asympsiphi{11}{23}+\asympsiphi{12}{24}-\asympsiphi{13}{21}
       -\asympsiphi{14}{22})\be_{3}f +\notag \\
   &\hspace*{3ex}
      ( \asympsiphi{11}{24}-\asympsiphi{12}{23}+\asympsiphi{13}{22}
       +\asympsiphi{21}{14})\be_{23}f \notag\\
   &= (\overline{\psi}_1 \phi_2+\overline{\psi}_2\phi_1)f 
    = \tr(\left[\begin{matrix}
              \overline{\psi}_1 & \overline{\psi}_2 
          \end{matrix}\right] 
          \left[\begin{matrix}
	      0 & 1\\1 & 0
	  \end{matrix}\right]
	  \left[\begin{matrix}
	     \phi_1\\ \phi_2
	  \end{matrix}\right])f
\label{eq:betaplus13}
\end{align}
where $s_1 = \be_1 \in \cb{M}$. Since the matrix is similar to a
diagonal matrix $\diag(1,-1)$, the automorphism group of this form is the
quaternionic unitary group $U_{1,1}\BH$. Then,
\begin{align}
  \beta_{-}(\psi,\phi)
    = s_2 \bar{\psi} \phi
   &= ( \asympsiphi{11}{21}-\asympsiphi{12}{22}-\asympsiphi{13}{23}
       +\asympsiphi{14}{24}) f + \notag \\
   &\hspace*{3ex}
      ( \asympsiphi{11}{22}+\asympsiphi{12}{21}+\asympsiphi{13}{24}
       +\asympsiphi{14}{23})\be_{2}f+\notag\\
   &\hspace*{3ex}
      ( \asympsiphi{11}{23}-\asympsiphi{12}{24}+\asympsiphi{13}{21}
       -\asympsiphi{14}{22})\be_{3}f+\notag \\
   &\hspace*{3ex}
      ( \sympsiphi{11}{24}+\sympsiphi{12}{23}-\sympsiphi{13}{22}
       -\sympsiphi{14}{21}) \be_{23}f \notag\\
   &= (\starc{\psi}_1 \phi_2-\starc{\psi}_2\phi_1)f
    = \tr(\left[\begin{matrix}
              \starc{\psi}_1 & \starc{\psi}_2 
	  \end{matrix}\right] 
	  \left[\begin{matrix}
	    0 & 1\\-1 & 0
	  \end{matrix}\right]
	  \left[\begin{matrix}
	    \phi_1\\ \phi_2
	  \end{matrix}\right])f
\label{eq:betaminus13}
\end{align}
where $s_2 = \be_1 \in \cb{M}$ and $\starc{\psi}_1,\starc{\psi}_2$ denotes an
anti-involution on $\BK \cong \BH$ which sends $\psi_{i1} + \psi_{i2} \be_{2}
+ \psi_{i3} \be_{3} + \psi_{i4} \be_{23} \mapsto \psi_{i1} + \psi_{i2} \be_{2}
+ \psi_{i3} \be_{3} - \psi_{i4} \be_{23}$. In~\cite[Tables 1 and 2, p. 236]{lounesto} we find that $\beta_{+}$ and $\beta_{-}$ scalar products are invariant under $Sp(2,2)$.\footnote{%
$Sp(2,2) = U(2,2) \cap Sp(4,\BC)$ and $Sp(2,2) /\{\pm 1 \} \cong SO_{+}(4,1)$.}
\label{example3}
\end{example}
\begin{remark}
If we define quaternions as pairs of complex numbers $\BH=\BC \oplus \BC j$ 
via the Cayley-Dickson doubling process~\cite[Sect. 23.2]{lounesto} where
$\BC=\BR \oplus \BR i,$ that is, write quaternions as $q=z+w j$ or as
$q=(z,w)$ where $z,w \in \BC$, then the star anti-involution
in~(\ref{eq:betaminus13}) is an $\BR$-linear map of $\BH$ sending
$z+w j \mapsto z+\bar{w} j $ or $(z,w)\mapsto (z,\bar{w})$ where $\bar{w}$
is the complex conjugate of $w$. Observe that the $\BR$-linear map of $\BH$
sending $z+w j \mapsto \bar{z}- w j$ or $(z,w)\mapsto (\bar{z},-w)$ is the
quaternionic conjugation of~$\BH$.
\end{remark}

We have classified all automorphism groups of the transposition scalar product
for quaternionic simple Clifford algebras. The results are summarized in 
Table~3. 

\subsection{Semisimple Clifford algebras}
\label{semisimple}
Faithful spinor representation of a semisimple Clifford algebra $\cl_{p,q}$
($p-q =1 \bmod 4$) is realized in a left ideal $\check{S} = S\oplus \hat{S}
= \cl_{p,q}e$ where $e = f + \hat{f}$ for any  primitive idempotent $f$.
Recall that $\hat{\phantom{u}}$ denotes grade involution. We refer
to~\cite[pp. 232--236]{lounesto} for some of the concepts. In particular, 
$S = \cl_{p,q}f$ and $\Sg = \cl_{p,q}\fg.$ Thus, every spinor $\check{\psi}
\in \check{S}$ has unique components $\psi \in S$ and $\psig \in \Sg$. We
refer to the elements $\check{\psi} \in \check{S}$ as ``spinors" whereas to
its components $\psi \in S$ and $\psig \in \Sg$ we refer as
``$\frac12$-spinors". 

For the semisimple Clifford algebras $\cl_{p,q},$ we will view spinors 
$\check{\psi} \in \check{S}= S \oplus \hat{S}$ as ordered pairs 
$(\psi,\psi_g) \in S \times \hat{S}$ when $\check{\psi} = \psi + \psi_g.$
Likewise, we will view elements $\check{\lambda}$ in the double fields
$\check{\BK} = \BK \oplus \hat{\BK}$ as ordered pairs $(\lambda,\lambda_g)
\in \BK \times \hat{\BK}$ when $\check{\lambda} = \lambda + \lambda_g$.
As before, $\BK = f \cl_{p,q}f$ while $\hat{\BK} = \hat{f} \cl_{p,q}\hat{f}$.
Recall that $\check{\BK} \cong \fpower{\BR}
\stackrel{\mbox{\scriptsize{def}}}{=} \BR \oplus \BR$ or
$\check{\BK} \cong \fpower{\BH} \stackrel{\mbox{\scriptsize{def}}}{=}
\BH \oplus \BH$ when, respectively, $p-q = 1 \bmod 8$, or $p-q = 5 \bmod 8$.

In this section we classify automorphism groups of the transposition scalar
product
\begin{gather}
  \check{S} \times \check{S} \rightarrow \check{\BK}, \quad 
  (\check{\psi},\check{\phi}) \mapsto 
  \tp(\check{\psi},\check{\phi}) \stackrel{\mbox{\scriptsize{def}}}{=} 
  (\tp(\psi)\phi, \tp(\psig)\phig) \in \check{\BK}
\label{eq:tsprod}
\end{gather}
when $\chpsi = \psi+\psig$ and $\chphi = \phi+\phig$.

\begin{proposition}
Let $\Gpqe{p}{q} \subset \cl_{p,q}$ where $\cl_{p,q}$ is a semisimple Clifford
algebra. Then, $\Gpqe{p}{q}$ is: The double orthogonal group
$\fpower{O(N)} \stackrel{\mbox{\scriptsize{def}}}{=} O(N) \times O(N)$ when
$\check{\BK} \cong \fpower{\BR}$ or the double compact symplectic group
$\fpower{Sp(N)} \stackrel{\mbox{\scriptsize{def}}}{=} Sp(N) \times Sp(N)$
when $\check{\BK} \cong \fpower{\BH}$.\footnote{%
Recall that $Sp(N)=U_{\BH}(N)$ where $U_{\BH}(N)$ is the quaternionic
unitary group~\cite{fultonharris}.}
That is,
\begin{equation}
  \Gpqe{p}{q}
    = \begin{cases} 
        \fpower{O(N)} = O(N) \times O(N)   & \textit{when $p-q =1 \bmod 8;$} \\
        \fpower{Sp(N)} = Sp(N) \times Sp(N) & \textit{when $p-q =5 \bmod 8;$}
      \end{cases}
\label{eq:semisimpleGpqe}
\end{equation}
where $N=2^{k-1}$ and $k=q-r_{q-p}$. 
\label{propsemisimpleGpqe}
\end{proposition}
\begin{proof}
This result follows directly when we recall that $\cl_{p,q} \cong
\fpower{\Mat(N,\BR)}$ or $\cl_{p,q} \cong \fpower{\Mat(N,\BH)}$ in the
semisimple case depending whether $p-q =1 \bmod 8$ or $p-q =5 \bmod 8$,
respectively. Then, apply Proposition~\ref{propsimpleGpqe} to each non
universal component of $\cl_{p,q}$ while remembering to reduce $N$ by a
factor of $2$.
\end{proof}

The automorphism groups $\Gpqe{p}{q}$ for semisimple Clifford algebras
$\cl_{p,q}$ for $p+q\leq 9$ are shown in Tables~4 and~5. All results in these
tables, like in Tables 1, 2, and 3, have been verified with $\mathtt{CLIFFORD}$.

{\small
\begin{table}[t]
\label{tab:t4}
\begin{center}
\renewcommand{\arraystretch}{1.4}
\begin{tabular}{|p{0.3in}|p{0.45in}|p{0.33in}|p{0.33in}|p{0.53in}|p{0.33in}|p{0.43in}|p{0.43in}|c|}
\multicolumn{9}{c}
{\bf Table 4: Automorphism group $\Gpqe{p}{q}$ of $\tp(\psi)\phi$}\\
\multicolumn{9}{c}
{\bf in semisimple Clifford algebras $\cl_{p,q} \cong \fpower{\Mat(2^{k-1},\BR)}$}\\
\multicolumn{9}{c}
{$k=q-r_{q-p}, $ $p-q = 1 \bmod 4, \, p-q = 1 \bmod 8$}\\\hline
\centering{$(p,q)$} & \centering{$(1,0)$} & \centering{$(2,1)$} & \centering{$(3,2)$} & \centering{$(0,7)$} & \centering{$(4,3)$} & \centering{$(1,8)$} & \centering{$(5,4)$} & $(9,0)$\\\hline
\centering{$\Gpqe{p}{q}$} & \centering{$\boxed{\fpower{O(1)}}$} & \centering{$\fpower{O(2)}$} & \centering{$\fpower{O(4)}$} & 
\centering{$\boxed{\boxed{\fpower{O(8)}}}$} & 
\centering{$\fpower{O(8)}$} & \centering{$\fpower{O(16)}$} & \centering{$\fpower{O(16)}$} & $\boxed{\fpower{O(16)}}$ \\\hline 
\end{tabular}
\end{center}
\end{table}
}%
{\small
\begin{table}[t]
\label{tab:t5}
\begin{center}
\renewcommand{\arraystretch}{1.4}
\begin{tabular}{|c|c|c|c|c|c|c|c|}
\multicolumn{8}{c}
{\bf Table 5: Automorphism group $\Gpqe{p}{q}$ of $\tp(\psi)\phi$}\\
\multicolumn{8}{c}
{\bf in semisimple Clifford algebras 
  $\cl_{p,q} \cong \fpower{\Mat(2^{k-1},\BH)}$}\\
\multicolumn{8}{c}
{$k=q-r_{q-p},$ $p-q = 5 \bmod 4, \, p-q =  \bmod 8$}\\\hline
$(p,q)$ & $(0,3)$ & $(1,4)$ & $(5,0)$ & $(2,5)$ & $(6,1)$ & $(3,6)$ & $(7,2)$\\\hline
$\Gpqe{p}{q}$ & $\boxed{\boxed{\fpower{Sp(1)}}}$ & $\fpower{Sp(2)}$ & 
$\boxed{\fpower{Sp(2)}}$ & $\fpower{Sp(4)}$ & $\fpower{Sp(4)}$ & $\fpower{Sp(8)}$ & 
$\fpower{Sp(8)}$ \\\hline
\end{tabular}
\end{center}
\end{table}
}%

We show two examples.
\begin{example}
Consider $\cl_{2,1} \cong \fpower{\Mat(2,\BR)}$ with $f=\frac14(1+\be_1)
(1+\be_{23})$ and $\hat{f} = \frac14(1-\be_1)(1+\be_{23})$. Then, a transversal
of the stabilizer $\Gpqf{2}{1}{f}=\Gpqf{2}{1}{\hat{f}} \lhd \Gpq{2}{1}$ is
$\cb{M} = \{1,\be_{2}\}$. Let $\psi,\phi \in S = \cl_{2,1}f$ and
$\psig,\phig \in \hat{S} = \cl_{2,1}\hat{f}$. Hence, 
\begin{alignat}{2}
  \psi
    &= \psi_{11}f + \psi_{21}\be_{2}f, 
    &  \phi 
    &= \phi_{11}f + \phi_{21}\be_{2}f, \\ 
  \psig
    &= \psi_{11}\hat{f} + \psi_{21}\be_{2}\hat{f}, \qquad 
    &  \phig 
    &= \phi_{11}\hat{f} + \phi_{21}\be_{2}\hat{f} 
\end{alignat}
where $\psi_{i1},\phi_{i1} \in \BR$, $i=1,2.$\footnote{%
By a small abuse of notation, we write the same numerical components for
$\psi$ and $\psig$, also for $\phi$ and~$\phig$, although in general they
may be different. In order to identify the automorphism groups what matters
is really the form of the polynomial in the numerical components but not the
actual components.}
It is then easy to see that the transposition scalar product on $S$:
\begin{gather}
  \tp(\psi)\phi
    = (\psi_{11}\phi_{11}+\psi_{21}\phi_{21})f
  \label{eq:firstform} 
\end{gather}
is invariant under $O(2)$. Likewise, the transposition scalar product on
$\hat{S}$
\begin{gather}
  \tp(\psig)\phig
    = (\psi_{11}\phi_{11}+\psi_{21}\phi_{21})\hat{f} 
\label{eq:secondform}
\end{gather}
is also invariant under $O(2)$. Thus, on $\check{S} = S \oplus \hat{S},$ we
get
\begin{gather}
  \tp(\check{\psi})\check{\phi}
    = (\tp(\psi)\phi,\tp(\psig)\phig)
    = (\psi_{11}\phi_{11}+\psi_{21}\phi_{21})(f,\hat{f}) \in \check{K}
\end{gather}
where $\check{K} \cong \fpower{\BR}$ and $\Gpqe{2}{1} = \fpower{O(2)}$. Of
course, in the faithful spinor representation in $\check{S} = S \oplus \hat{S}$
over the double field $\fpower{\BR}$, we have
\begin{gather}
  [\tp(\check{\psi})]
    = \left[\begin{matrix}
       (\psi_{11},\psi_{11}) & (\psi_{21},\psi_{21})\\
       (0,0) & (0,0)
      \end{matrix} \right] 
      \quad 
      \mbox{and} 
      \quad
  [\check{\phi}]
    = \left[\begin{matrix}
       (\phi_{11},\phi_{11}) & (0,0)\\
       (\phi_{21},\phi_{21}) & (0,0)
      \end{matrix}\right] 
\end{gather}
which gives the same result in matrix form:
\begin{gather}
  [\tp(\check{\psi})][\check{\phi}]
    = \left[\begin{matrix}
       (\lambda, \lambda_g) & (0,0)\\
       (0,0) & (0,0)
      \end{matrix}\right] 
\end{gather}
where $\lambda_g =\lambda=\psi_{11}\phi_{11}+\psi_{21}\phi_{21} \in \BR$. For
a comparison, notice that the scalar products $\beta_{+}$ and $\beta_{-}$ on
$\check{S}$ are
\begin{align}
  \beta_{+}(\check{\psi},\check{\phi}) 
    &= (0,0),\\
  \beta_{-}(\check{\psi},\check{\phi})
    &= (\psi_{11}\phi_{21}-\psi_{21}\phi_{11})(f,\hat{f})
\end{align}
and their automorphism groups are, respectively, $\fpower{GL(2,\BR)}$ and  
$\fpower{Sp(2,\BR)}$.\footnote{%
In~\cite[Table 1, page 236]{lounesto}, Lounesto shows for short only
$GL(2,\BR)$ whereas in~\cite[Table 2, page 236]{lounesto}, he also shows 
$\fpower{Sp(2,\BR)}$. Thus, our definition of the symplectic group
$Sp(N,\BR)$ is the same as his.}
\end{example}

\begin{example}
Consider $\cl_{1,4} \cong \fpower{\Mat(2,\BH)}$ with $f = \frac14
(1+\be_{234})(1+\be_{15})$ and $\hat{f} = \frac14(1-\be_{234})(1+\be_{15})$.
Then, a transversal of the stabilizer $\Gpqf{1}{4}{f} = \Gpqf{1}{4}{\hat{f}}
\lhd \Gpq{1}{4}$ is $\cb{M}=\{1,\be_{1} \}$. Let $\psi,\phi \in S=\cl_{1,4}f$
and $\psig,\phig \in \hat{S}=\cl_{1,4}\hat{f}$. Let's assign, modulo~$f$,
quaternionic components to $\frac12$-spinors as follows:
$$
  \psi_i \colonequal \psi_{11} + \psi_{12}\be_{2} + \psi_{13}\be_{3} 
             + \psi_{14}\be_{23}, \quad
  \phi_i \colonequal \phi_{11} + \phi_{12}\be_{2} + \phi_{13}\be_{3}
             + \phi_{14}\be_{23}, 
$$
for $i=1,2$. Thus, 
\begin{alignat}{2}
  \psi
    &= f\psi_{1} + \be_{1}f\psi_{2},
    &  \phi 
    &= f\phi_{1} + \be_{1}f\phi_{2}, \\ 
  \psig
    &= \hat{f}\psi_{1} + \be_{1}\hat{f}\psi_{2}, \qquad 
    &  \phig
    &= \hat{f}\phi_{1} + \be_{1}\hat{f}\phi_{2}. 
\end{alignat}
Then, modulo~$f$, we find the transposition scalar product on $S$
\begin{gather}
  \tp(\psi)\phi
    = (\overline{\psi}_1 \phi_{1} + \overline{\psi}_2 \phi_{2})f
\end{gather}
where $\overline{\psi}_i$ denotes ``quaternionic'' conjugation. It is easy to
see that this product is invariant under $Sp(2)=U_\BH(2)$. Furthermore,
modulo~$\hat{f}$, we find the transposition scalar product on~$\hat{S}$
\begin{gather}
  \tp(\psig)\phig
    = (\overline{\psi}_1 \phi_{1} + \overline{\psi}_2 \phi_{2}) \hat{f}
\end{gather}
hence it is also invariant under $Sp(2)$. Therefore, on $\check{S} =
S\oplus\hat{S}$, we get
\begin{gather}
  \tp(\check{\psi})\check{\phi}
    = (\tp(\psi)\phi,\tp(\psig)\phig)
    = (\overline{\psi}_1 \phi_{1} + \overline{\psi}_2 \phi_{2})(f,\hat{f}) 
       \in \check{\BK} \cong \fpower{\BH}   
\end{gather}
and $\Gpqe{1}{4} = \fpower{Sp(2)}$. We can also verify this result using the
faithful spinor representation in $\check{S}=S \oplus \hat{S}$ over the double
(skew) field $\fpower{\BH}$:
\begin{gather}
  [\tp(\check{\psi})]
    =\left[\begin{matrix} 
      (\overline{\psi}_{1},\overline{\psi}_{1}) & 
      (\overline{\psi}_{2},\overline{\psi}_{2}) \\
      (0,0) & 
      (0,0)
     \end{matrix}\right]
     \quad 
     \mbox{and}
     \quad
  [\check{\phi}]
    =\left[\begin{matrix}
      (\phi_{1},\phi_{1}) & (0,0)\\
      (\phi_{2},\phi_{2}) & (0,0)
     \end{matrix}\right] 
\end{gather}
which gives the same result in matrix form:
\begin{gather}
  [\tp(\check{\psi})][\check{\phi}]
    = \left[\begin{matrix}
        (\lambda, \lambda_g) & (0,0)\\
        (0,0) & (0,0)
      \end{matrix}\right] 
\end{gather}
where $\lambda_g =\lambda=\overline{\psi}_{1}\phi_{1} + \overline{\psi}_{2}
\phi_{2} \in \BH$. We again compare this product with $\beta_{+}$ and
$\beta_{-}$ on $\check{S}$ which are:
\begin{align}
  \beta_{+}(\check{\psi},\check{\phi})
    &= (\overline{\psi}_1 \phi_2 + \overline{\psi}_2 \phi_1)(f,\hat{f}),\\
  \beta_{-}(\check{\psi},\check{\phi})
    &= (0,0), 
\end{align}
and their automorphism groups are, respectively, $\fpower{U_{1,1}\BH} =
U_{1,1}\BH \times U_{1,1}\BH$ and $\fpower{GL(2,\BH)}.$\footnote{%
The group $U_{p,q}\BH$ is defined in~\cite[Page 99]{fultonharris} as the
group of automorphisms of a \textit{Hermitian form} of signature $(p,q)$ on
a quaternionic vector space $V$ of dimension $p+q$ thus having the standard
expression $\sum_{i=1}^p \overline{v}_i w_i - \sum_{i=p+1}^{p+q} 
\overline{v}_i w_i.$ In~\cite[Table 1, page 236]{lounesto} we find
$\fpower{Sp(2,2)}$ and in~\cite[Table 2, page 236]{lounesto} we find 
$GL(2,\BH)$, respectively.}
\end{example}

\medskip\section{Clifford algebra as a twisted group ring}
\label{groupalgebra}

For any group $G$, its commutator subgroup $G'$ is a normal subgroup of~$G$,
and the quotient group $G/G'$ is abelian.~\cite[Prop. 5.57]{rotman} Hence,
since the vee group $\Gpq{p}{q}$ is not abelian (for $n=p+q\ge 2$) of order
$2\cdot 2^{p+q}$ and its commutator subgroup $\Gpq{p}{q}' = \{ \pm 1\},$ the
abelian quotient group is just
$$
  \Gpq{p}{q}/\Gpq{p}{q}' \cong (\BZ_2)^n
    = \underbrace{\BZ_2 \times \cdots \times \BZ_2}_{n} 
      \quad \mbox{with} \quad n=p+q.
$$
Thus, we can view the universal Clifford algebra $\cl_{p,q}$ of a
non-degenerate quadratic form $Q$ of signature $(p,q)$ as a
\textit{twisted group ring} $\cl_{p,q} = \BR^t[(\BZ_2)^n]$ over the abelian
group $(\BZ_2)^n.$\footnote{%
For a general theory of group rings see~\cite{passman}.} 
Let $e_1,\ldots,e_n$ be generators for $n$ isomorphic copies of $\BZ_2$, that
is,
$$
  \BZ_2 \cong \langle e_1 \rangle \cong \langle e_2 \rangle 
        \cong \cdots \cong  \langle e_n \rangle
$$
where the cyclic groups $\langle e_i \rangle = \{e_i^{a_i}\}$, $a_i=0,1$,
$i=1,\ldots,n,$ are written multiplicatively.\footnote{%
We could simplify notation by identifying an ordered $n$-tuple 
$(e_1^{a_1},e_2^{a_2},\ldots, e_n^{a_n})$ from $(\BZ_2)^n$ with a
\textit{monomial} $e_1^{a_1} e_2^{a_2}\cdots e_n^{a_n}$ in $\BR^t[(\BZ_2)^n]$.}
Then, the \textit{twisted group ring} $\BR^t[(\BZ_2)^n]$~\cite[Sect. 2]{passman}
is an associative $\BR$-algebra with basis $\{\bar{x} \mid x \in (\BZ_2)^n\}$
and multiplication defined distributively for all $x,y \in (\BZ_2)^n$ as
$$ 
  \bar{x} \bar{y} 
    = \gamma(x,y)\, \overline{xy}, 
  \qquad \gamma(x,y) \in \BR^{*} 
    = \BR \setminus \{0\}.
$$
Due to the required associativity of the algebra
$(\bar{x} \bar{y}) \, \bar{z} = 
\bar{x} \, (\bar{y} \, \bar{z})$ for any $x,y,z \in (\BZ_2)^n,$ the function
$\gamma: (\BZ_2)^n \times (\BZ_2)^n \rightarrow \BR^{*}$ must satisfy the
relation
$$
  \gamma(x,y)\gamma(xy,z) 
     = \gamma(y,z) \gamma(x,yz), 
       \qquad \forall z,y,z \in (\BZ_2)^n
$$
which implies that $\gamma$ is a $2$-cocycle. Furthermore, it is
shown~\cite[Lemma 2.1]{passman} that a general twisted group ring $K^t[G]$ has
an identity $\gamma(1,1)^{-1}\bar{1}$ and that the elements $\bar{x} \in 
K^t[G]$ are all units. The inverse of $\bar{x}$ is:
$$
  \bar{x}^{-1}
    = \gamma(x,x^{-1})\gamma(1,1)^{-1} \overline{x^{-1}}
    = \gamma(x^{-1},x)\gamma(1,1)^{-1} \overline{x^{-1}}. 
$$

One can easily show that the formula $(xy)^{-1} = y^{-1}x^{-1}, \forall x,y
\in G,$ carries over to $(\bar{x}\bar{y})^{-1} = \bar{y}^{-1}\bar{x}^{-1},
\forall \bar{x},\bar{y} \in K^t[G]$: 
\begin{align}
  (\bar{x}\bar{y})(\bar{y}^{-1}\bar{x}^{-1})
    &= \bar{x}(\bar{y} \bar{y}^{-1}) \bar{x}^{-1}
     = \bar{x} \gamma(1,1)^{-1} \bar{1} \bar{x}^{-1}
     = \bar{x}\bar{x}^{-1}\gamma(1,1)^{-1}\bar{1}\notag\\
    &= \gamma(1,1)^{-1}\bar{1}\gamma(1,1)^{-1}\bar{1}
     = \gamma(1,1)^{-1}\gamma(1,1)^{-1}\bar{1}\bar{1}\notag\\
    &= \gamma(1,1)^{-1}\gamma(1,1)^{-1}\gamma(1,1)\bar{1}
     = \gamma(1,1)^{-1}\bar{1} 
\end{align}

In~\cite[Page 5]{passman}, one finds this definition of a
map $*: K[G] \rightarrow K[G]$ in a general group ring $K[G]$:
$$
  \left(\sum a_x x\right)^* = \sum a_x x^{-1} 
$$
where the summation is over $x \in G$. It follows easily that, for any
$\alpha,\beta \in K[G],$ we have: 
(i) $(\alpha + \beta)^* =  \alpha^* + \beta^*$, 
(ii) $(\alpha \beta)^{*} = \beta^{*}\alpha^{*}$,
(iii) $\alpha^{**} = \alpha$, 
(iv) $1^* = 1$ for the unity. 
Hence, $*$ is an anti-involution of the group ring. 
  
Notice that a similar map can be defined in any \textit{twisted} group ring 
$*:K^t[G] \rightarrow K^t[G]$ namely
$$
  \left(\sum a_x \bar{x}\right)^* = \sum a_x \bar{x}^{-1}. 
$$
Furthermore, this map is an anti-involution of $K^t[G]$ as it satisfies the
same four properties (i)--(iv). In particular, property (ii) follows from
the identity $(\bar{x}\bar{y})^{-1} =\bar{y}^{-1}\bar{x}^{-1}$ shown above.

We recall from~\cite[Lemma 1]{part2} some properties of the transposition
anti-involution $\tp$, and, in particular, its action $\tp(m)=m^{-1}$ on a
monomial $m$ in the Grassmann basis $\cb{B}$ which is, as we see now,
identical to the action $*(m)=m^{-1}$ on every $m \in \cb{B}$. We formulate
our concluding result.
\begin{theorem}
\label{theorem2}
The anti-involution $\tp$ on the Clifford algebra $\cl_{p,q}$ is the
anti-involution $*$ of $\cl_{p,q}$ viewed as the twisted group ring
$\BR^t[(\BZ_2)^n].$ 
\end{theorem}
For a Hopf algebraic discussion of Clifford algebras as twisted group
algebras, see~\cite{albuquerquemajid,morier-genoudovsienko} and references
therein.

\begin{remark}
Following~\cite[Chapt. 21]{lounesto}, consider $n$-tuples $\underline{a} = 
a_1a_2\ldots a_n$ of binary digits $a_i =0,1,$ with addition $\underline{a}
\oplus \underline{b} = \underline{c}$ defined by term wise addition modulo~$2$.
This is the natural component wise product group structure on $(\BZ_2)^n$.
Thus, with this operation, the set of $n$-tuples is a group and the group
characters are \textit{Walsh functions}
$$
  w_{\underline{a}}(\underline{b}) = (-1)^{\sum_{i=1}^{n} a_ib_i}
$$
As group characters, the Walsh functions satisfy $w_{\underline{k}}
(\underline{a} \oplus \underline{b}) = w_{\underline{k}}(\underline{a}) 
w_{\underline{k}}(\underline{b})$. For any binary $n$-tuple 
$\underline{a}=a_1a_2 \ldots a_n$, let's adopt notation 
$$
  e_{\underline{a}} \stackrel{\mbox{\scriptsize{def}}}{=} 
     e_1^{a_1}e_2^{a_2}\cdots e_n^{a_n}
$$
for any element $(e_1^{a_1},e_2^{a_2},\ldots, e_n^{a_n}) \in (\BZ_2)^n$.
Then, Lounesto shows how one can define the Clifford product in the algebra
$\cl_{p,q}$ on the ``basis monomials'' (later extended by linearity and
associativity to all algebra elements) with the help of Walsh functions,
namely,
\begin{gather}
  e_{\underline{a}} e_{\underline{b}}  \stackrel{\mbox{\scriptsize{def}}}{=}
    (-1)^{\sum_{i=1}^{p} a_i b_i} w_{\underline{a}}(h(\underline{b})) 
      e_{\underline{a} \oplus \underline{b}}
\label{eq:walshproduct}
\end{gather}
where $h$, the inverse of the so called \textit{Gray code}, is defined as 
$$
  h(\underline{a})_i= \sum_{j=1}^{i} a_j \bmod 2.
$$ 
Thus, we have a close connection between viewing Clifford algebra $\cl_{p,q}$
as a twisted group ring $\BR^t[(\BZ_2)^n]$ and as a real associative 
algebra with the monomial basis $\{e_{\underline{a}}\}$ for all binary
$n$-tuples $\underline{a}$ and multiplication defined with the help of the
Walsh function and inverse Gray code as in~(\ref{eq:walshproduct}). In
particular, one should be able to express the twisting $2$-cocycle $\gamma$
needed to define the product in $\BR^t[(\BZ_2)^n]$ 
(see~\cite{albuquerquemajid,morier-genoudovsienko}) in terms of the
Walsh functions and the Gray code.   
\end{remark}

\medskip\section{Conclusions}
\label{conc}

We summarize and conclude our work on the transposition involution in real
Clifford algebras, including the developments of~\cite{part1,part2}. The present
set of three papers develops in a concise and explicit manner all important
aspects of the transposition map, which we have lifted from the matrix
representations to an abstract map in real Clifford algebras. We did this in an
exhaustive manner for all signatures $\varepsilon=(p,q)$ for $p+q\leq 9$ well
aware of the fact that with $\bmod\, 8$ periodicity of real Clifford algebras,
we have a complete treatment for any signature. Our goal was to explicitly
verify all our claims with \texttt{CLIFFORD}~\cite{ablamfauser2009}. This has
revealed some corrections necessary to results already published in the
literature. We also want to mention that despite an extensive search, we were
unable to find our transposition map in the literature, contrary to
the immediate feeling that this map `should already have been known'. However,
see its relation to the `*-map' in the twisted group rings given in  
Theorem~\ref{theorem2}. Of course, the transposition map for Euclidean and anti-Euclidean
signatures coincides with, respectively, the reversion and the conjugation
anti-involutions of $\cl_{p,q}$ used in the standard constructions of $\Spin$
and $\Pin$ groups, but for non-definite signatures it differs.

The first paper~\cite{part1} is concerned with a general setup of the
transposition map. This includes a careful choice of sorted bases, with an
admissible order, allowing for further generalizations. Our starting point was
to find the abstract map in abstract real Clifford algebras induced by the
linear algebra map of transposition in the left (right) regular matrix
representations of the real Clifford algebras. We do not reiterate the points
given in the conclusions of~\cite{part1} but we focus here on the general
outline of our work.

The transposition map explains the block structure (grading) of the regular representation matrices. Moreover, we exhibited an explicit decomposition of the transposition map with respect to the graded decomposition of the Clifford algebras $\cl_{p,q} = \cl_{p,0} \hotimes \cl_{0,q}$ for arbitrary signatures. This fact, depending on the symmetry of the employed non singular
polar bilinear form $B_{p,q}$, emerges from the decomposition of the related quadratic
form $Q_{p,q}= Q_{p,0}^\prime \perp Q_{0,q}^{\prime\prime}$ into a positive and a negative definite form. On the technical side we have extended the duality result of Lounesto
\begin{align}
  \langle u \JJ_{B} v , w\rangle &= \langle v, \tilde{u}\w v\rangle
\end{align}
relating the contraction and the Grassmann wedge product, to the general Clifford case
\begin{gather}
\langle u , vw\rangle = \langle \tilde{v}u, v\rangle 
\qquad\textrm{provided ~~~$\det Q_{p,q}\not=0$}\, .
\end{gather}
Choosing a sorted basis $\cb{B}$ which diagonalizes $Q_{p,q}$, we have considered two canonical isomorphisms $V\rightarrow V^*$ identifying the vector space $V$ with its dual~$V^*$. These maps are given by: (i) the dual basis $\cb{B}^*$ defined as $\e_i^*(\e_j)=\delta_{i,j}$ and (ii) the
reciprocal basis $\cb{B}^\flat$ defined as $\e_i^\flat = \dfrac{\e_i}{\varepsilon_i}$. This leads to the universal dual (or \textit{reciprocal}) Clifford algebra $\cl^*_n$. The 
isomorphism $\flat : V \rightarrow V^*$ depends on the quadratic form $Q_{p,q}$ or, equivalently, on its symmetric polar bilinear form $B_{p,q}$. The consequences are 
somewhat straightforward and we gave a few explicit examples.

In general, the regular representation is reducible. Thus, it is much more
interesting to explore how the transposition map is related to the irreducible
spinor representations. This is done in~\cite{part2}. Algebraic spinor spaces
$S$, seen as left $\cl$-modules, are minimal left ideals. These ideals can be
generated by a primitive idempotent $f$ such that $S(f)=\cl_{p,q}f$. As real
Clifford algebras are simple or are direct sums of two simple subalgebras, such
spinor ideals are right modules over $\BK$ which is either a (skew) field or a
double (skew) field. The spinor ideal ${}_{\cl}S(f)_\BK$ is hence a left $\cl$
and right $\BK$-module. The abstract transposition map factors into a
transposition map on the left action by $\cl$ and an (anti) involution on $\BK$.
It was shown, for simple and semisimple algebras, that this involution,
depending on the signature $\bmod \, 8$, amounts to the identity on $\BR$,
complex conjugation on $\BC$, and quaternionic conjugation on $\BH$. An
important role in these constructions is played by the 
vee group $\Gpq{p}{q}$ and the stabilizer group $\Gpqf{p}{q}{f}$ of the primitive
idempotent element $f$ generating the spinor space $S(f)$. We have explicitly
constructed left transversals of cosets which provide bases for the spinor
spaces. A full classification of the stabilizer groups $\Gpqf{p}{q}{f}$ was
given in~\cite{part2} for all signatures $p+q=n\leq 9$, which is complete when
taken $\bmod\, 8$ in both simple and semisimple algebras. These stabilizer
groups do play an important role in quantum computing too. The  latter, however,
usually employs only complex Clifford algebras which posses a much simpler 
$\bmod\, 2$ structure. 

We showed that the transposition map allows one to form a new transposition
scalar product on spinor spaces which we have studied in the present paper.
Only in the Euclidean and anti-Euclidean cases, this scalar product is
identical to the two known spinor scalar products, see for 
example~\cite{lounesto}, and it is different in all other signatures. The
present paper provides a complete classification of the invariance 
groups $\Gpqe{p}{q}$ of the transposition scalar product. 
We also gave a detailed exposition on $\Tpqf{p}{q}{f}$, 
the idempotent group $\Tpqf{p}{q}{f}$ of~$f$, and the field group $\Kpqf{p}{q}{f}$ of~$f$, 
as normal subgroups in the stabilizer group $\Gpqf{p}{q}{f}$ of the primitive idempotent $f$, and on their coset spaces $\Gpq{p}{q}/\Tpqf{p}{q}{f}$, $\Gpqf{p}{q}{f}/\Tpqf{p}{q}{f}$, and
$\Gpq{p}{q}/\Gpqf{p}{q}{f}$. These subgroups allow to construct very
effectively non-canonical transversals and hence basis elements of the spinor
spaces and the (skew double) field underlying the spinor space. This leads
to a description of the situation tabulated for $p+q=n\leq 9$ which is
complete and sufficient due to the $\bmod\, 8$ periodicity. A further aspect
which we have emphasized is the fact that the Clifford algebras can be seen
as twisted group rings $\BR^t[(\BZ_2)^n]$. In particular, we have observed
that our transposition is then mapped to the `star map' in such
rings~\cite{passman}. This observation relates our work to recent approaches to
Clifford algebras using Hopf algebraic methods applied to the twisted group 
ring $\BR^t[(\BZ_2)^n]$~\cite{albuquerquemajid,morier-genoudovsienko}. In
particular, the `star' map of a twisted group ring is actually the antipode map.

Let $\beta$ be an automorphism of $\BK$. The right module $S_{\BK}$ can be
given a new module structure $S_{\cdot\BK}$ by $\psi\cdot \lambda = 
\psi\beta(\lambda)$. Similarly we can twist the left action of $\cl$ seen
as the matrix algebra over $\BK$. As an example, the complex left module
${}_{\BC}M$ can receive an antilinear complex vector space structure
${}_{\BC\cdot}M$, such that $z\cdot m = \overline{z}m$. Seen as
representations, these modules are in general not isomorphic.

As a final remark, observe that the transposition anti-involution connects 
left and right spinor modules ($\tp :{}_{\cl}S_\BK \rightarrow {}_{\BK\cdot} 
S_{\cdot\cl}$) or, equivalently, regular left and right representations, and it 
implements an \emph{isomorphism} between the left and the right modules. 
%
%
Such a map introduces a Frobenius algebra structure~\cite{kadison}, in general 
twisted (or of the second kind), on the Clifford algebra and its respective representations. 
This is expected due to the fact that Clifford algebras can
be seen as twisted group rings. Our work on the stabilizer and invariance
groups shows that only in the case when $\BK \cong \BR$ this identification is
straightforward. In all other cases, the left and the right modules have
different actions of the base field $\BK$ due to the action of the
transposition involution on the base field, which induces the twisting in
the sense of the twisted Frobenius structures~\cite[Ch. 7.1]{kadison}. The
clarifications brought forward in this set of three articles will help to
investigate this fact.


\end{document}